\documentclass[letterpaper]{article} 
\usepackage{aaai2026}  
\usepackage{times}  
\usepackage{helvet}  
\usepackage{courier}  
\usepackage[hyphens]{url}  
\usepackage{graphicx} 
\usepackage{natbib}  
\usepackage{caption} 
\usepackage{algorithm}

\usepackage{newfloat}
\usepackage{listings}
\DeclareCaptionStyle{ruled}{labelfont=normalfont,labelsep=colon,strut=off} 
\floatstyle{ruled}
\newfloat{listing}{tb}{lst}{}
\floatname{listing}{Listing}
\usepackage{amsfonts}
\usepackage{mathtools}
\usepackage{algpseudocode}
\usepackage{booktabs}
\usepackage{multirow}
\usepackage{placeins}  
\usepackage{amsthm} 
\usepackage{thmtools}

\newtheorem{theorem}{Theorem}
\newtheorem{lemma}{Lemma}

\newtheorem{definition}{Definition}
\newtheorem{assumption}{Assumption}

\newcommand{\BibTeX}{B\kern-.05em{\sc i\kern-.025em b}\kern-.08em\TeX}

\newcommand{\fref}[1]{\fref{#1}}

\newcommand{\cref}[1]{Chapter~\ref{#1}}

\title{Boltzmann-based Exploration for Robust Decentralized Multi-Agent Planning (Extended Version)}

\author {
    Nhat D. A. Nguyen\textsuperscript{\rm 1,2},
    Duong D. Nguyen\textsuperscript{\rm 1},
    Gianluca Rizzo\textsuperscript{\rm 3},
    Hung X. Nguyen\textsuperscript{\rm 1}
}
\affiliations {
    \textsuperscript{\rm 1}Adelaide University, Australia\\
    \textsuperscript{\rm 2}The University of Foggia, Italy\\
    \textsuperscript{\rm 3}HES-SO Valais, Switzerland, and The University of Turin, Italy\\
    nhatdaoanh.nguyen@adelaide.edu.au, duong.nguyen@adelaide.edu.au, gianluca.rizzo@hevs.ch, hung.nguyen@adelaide.edu.au
}

\begin{document}

\maketitle

\begin{abstract}

Decentralized Monte Carlo Tree Search (Dec-MCTS) is widely used for cooperative multi-agent planning but struggles in sparse or skewed reward environments. We introduce Coordinated Boltzmann MCTS (CB-MCTS), which replaces deterministic UCT with a stochastic Boltzmann policy and a decaying entropy bonus for sustained yet focused exploration. While Boltzmann exploration has been studied in single-agent MCTS, applying it in multi-agent systems poses unique challenges. CB-MCTS is the first to address this. We analyze CB‑MCTS in the simple-regret setting and show in simulations that it outperforms Dec-MCTS in deceptive scenarios and remains competitive on standard benchmarks, providing a robust solution for multi-agent planning.
\end{abstract}


\section{Introduction}

Decentralized Monte Carlo Tree Search (Dec-MCTS) is an increasingly popular paradigm for cooperative multi-agent planning \cite{best2019dec, li2019integrating, nguyen2024united}. Its anytime performance, domain-agnostic design, and online replanning capabilities make it well-suited for applications requiring scalability, fast response, and coordination across distributed agents, such as information gathering, precision farming, and networked robotics \cite{claes2017decentralised, sukkar2019multi, nguyen2024decentralized}.

Current Dec-MCTS algorithms rely on the Upper Confidence Bound applied to Trees (UCT) and its variants \cite{kocsis2006improved} to guide the search process. UCT selects actions according to the principle of \textit{optimism in the face of uncertainty}, prioritizing branches with high empirical rewards. While this mechanism is effective when rewards are smooth or moderately stochastic \cite{munos2014bandits}, its efficiency degrades in skewed, sparse, or deceptive reward landscapes. In such cases, early high-reward samples can mislead the search, causing the algorithm to overcommit to suboptimal branches while neglecting deeper paths that lead to higher rewards \cite{coquelin2007bandit, ramanujan2011trade, james2017analysis}. Although extensively studied in single-agent MCTS, the implications for decentralized multi-agent planning, where coordination amplifies the problem, remain largely unexamined.

This paper provides the first simple regret analysis of Dec‑MCTS in deceptive multi‑agent trees. We then introduce \emph{Coordinated Boltzmann Monte Carlo Tree Search} (CB‑MCTS), a distributed algorithm that addresses these limitations. CB-MCTS replaces the deterministic UCT selection with a stochastic Boltzmann policy and incorporates a decaying entropy-based bonus to sustain exploration while progressively focusing on high-value actions. Coordination among agents is achieved through a marginal contribution function that aligns each agent's local decisions with the global objective, mitigating the variance introduced by simultaneous actions. This approach enables the search to explore deceptive or initially suboptimal regions effectively, improving convergence to globally optimal strategies.

While Boltzmann exploration has been applied in single-agent MCTS \cite{cesa2017boltzmann, painter2023monte}, CB-MCTS is, to our knowledge, the first to adapt it to multi-agent planning. We show theoretically that CB‑MCTS achieves exponentially faster decay of simple regret than D‑UCT‑based Dec‑MCTS in deceptive trees. Extensive simulations demonstrate that CB-MCTS matches state-of-the-art methods on standard benchmarks while significantly outperforming them in scenarios with skewed or sparse reward distribution. Overall, CB-MCTS offers a robust and adaptable framework for multi-agent planning problems across smooth to sparse reward environments.

\section{Problem Statement}

We consider a cooperative multi-agent planning problem with $N$ agents in a shared environment modeled as an undirected graph $G=(V, E)$, where vertices are states and edges are actions. Each agent $n$ selects a valid action sequence $a^n \in \mathcal{A}^n$, subject to a cost budget $b(a^n) \le B^n$. The global objective $g(a)$ depends on the joint action $a=(a^1,\dots,a^N)$, and the goal is to maximize $g(a)$ within a planning budget $T$.
Dec-MCTS addresses this by letting each agent build its own search tree using repeated simulations. A trajectory from the root corresponds to a candidate action sequence, and nodes are selected using Discounted UCT (D-UCT), which weights empirical rewards by a discount factor $\gamma$.

While cumulative regret is standard in MCTS evaluation, in multi-agent planning with finite planning budgets, only the executed actions contribute to real-world outcomes. Hence, \emph{simple regret} $r_T = \mu^* - \mu_{J_T}$, the expected loss of executing the recommended action $J_T$ after $T$ planning iterations, is a more relevant metric \cite{tolpin2012mcts}. While UCT is guaranteed to converge to the optimal trajectory in the limit \cite{kocsis2006improved}, this convergence can be extremely slow in sparse, skewed, or deceptive reward structures, such as the classical D‑chain tree \cite{coquelin2007bandit,james2017analysis}. We show that this D‑chain pathology extends to Dec-MCTS with D‑UCT, where the difficulties are magnified, particularly for multi-agent coordination.

\begin{definition}
\label{def:d_chain}
    The multi-agent D-chain problem is an $M$-ary tree of depth $D$. At depth $d<D$, action 1 progresses to the next level; all others terminate with reward $(D-d)/D$. At depth $D$, action 1 gives a reward of 1; others give 0. Agents initially select non-progressing actions, and overlapping action sequences among agents do not yield extra reward.
\end{definition}

\begin{lemma}
    For a fixed $\gamma$, there exists a value $D$ such that Dec-MCTS with D-UCT fails to identify the optimal action sequence in the D-Chain problem.
\end{lemma}

As D-UCT is designed to minimize cumulative regret, it is unsuitable for environments that require extensive exploration (i.e., minimizing simple regret) since both regrets cannot be minimized simultaneously \cite{bubeck2011pure}. We formally bound the simple regret of Dec-MCTS with D-UCT as follows:

\begin{theorem}
The simple regret of Dec-MCTS with D-UCT is bounded by $\mathbb{E}[r_T] \leq  C \exp \left( - k \sqrt{T} \log T \right)$ for some constants $C, k > 0$.
\end{theorem}
The proofs of Lemma 1 and Theorem 1 are in the Appendix.

\section{Coordinated Boltzmann MCTS}

\subsection{Distributed CB-MCTS with Discounted Backup}

We propose Coordinated Boltzmann Monte Carlo Tree Search (CB-MCTS), a distributed algorithm for cooperative multi-agent planning. Each agent $n$ independently runs CB-MCTS over its search tree $\mathcal{T}^n$, where nodes represent states and edges represent actions. A root-to-leaf branch encodes a feasible action sequence. All agents aim to maximize the global utility $g$ through decentralized coordination.

To coordinate without centralization, each agent maintains a compressed representation of its tree consisting of (i) a subset $\hat{\mathcal{A}}^n$ of high-value rollouts and (ii) a probability mass function $p^n$ over these rollouts. The subset $\hat{\mathcal{A}}^n$ is obtained by selecting leaf nodes with the highest discounted empirical returns every $c$ iterations. The probabilities $p^n$ are updated via a decentralized gradient-based consensus protocol~\cite{best2019dec}, enabling each agent to form beliefs about others’ future trajectories without exchanging full trees.

The tree $\mathcal{T}^n$ is grown iteratively using the standard 4-step MCTS process~\cite{kocsis2006improved}. During \textit{selection}, the algorithm recursively chooses a child using the stochastic Boltzmann policy, stopping when encountering an unvisited child or reaching the planning horizon. The chosen unvisited child is then \textit{expanded}. A \textit{rollout} phase follows, where random actions are sampled until the horizon is reached.

\begin{algorithm}[tb]
    \caption{Overview of CB-MCTS for agent $n$}
    \label{alg: CB-MCTS}
    \begin{algorithmic}[1]
        \Require $g$, $T$, $c$, $B$, $\hat{\mathcal{A}}^{-n}$, $p^{-n}$
        \Ensure best action sequence $a^n$ for agent $n$
        \State $\mathcal{T}^n \leftarrow$ Initialize the search tree
        \State $t \leftarrow 0$
        \While{$t < T$}
            \If{$t \bmod c == 0$}
                \State $\hat{\mathcal{A}}^n \leftarrow$ Tree Compression $(\mathcal{T}^n)$
            \EndIf
            \For{a fixed number of iterations}
                \State $i \leftarrow$ Boltzmann Selection $(\mathcal{T}^n)$
                \State $[j, \mathcal{T}^n] \leftarrow$ Tree Expansion $(\mathcal{T}^n, i, h)$
                \State $a^n \leftarrow$ Simulation $(j, B)$
                \State $a^{-n} \leftarrow$ Sample $(\hat{\mathcal{A}}^{-n}, p^{-n})$
                \State $r(a^n) \leftarrow$ Marginal Contribution $(g, a^n, a^{-n})$
                \State $\mathcal{T}^n \leftarrow$ Backpropagation $(\mathcal{T}^n, j, r(a^n))$
                \State $t \leftarrow t + 1$
            \EndFor
            \State $[\hat{\mathcal{A}}^{-n}, p^{-n}] \leftarrow$ Update and Communicate $(\hat{\mathcal{A}}^n, p^n)$
        \EndWhile
        \State \Return $a^n \leftarrow \underset{a \in \hat{\mathcal{A}}^n}{\arg\max} [p^n(a)]$ 
    \end{algorithmic}
\end{algorithm}

To evaluate its rollout $a^n$, agent $n$ samples joint actions $a^{-n}$ for other agents from $(\hat{\mathcal{A}}^{-n}, p^{-n})$ and computes its marginal contribution:
\begin{equation}\label{eq: marginal}
    r(a^n) = g(a^n, a^{-n}) - g(a^{-n}),
\end{equation}
which aligns each agent's objective with the global utility while mitigating variance in multi-agent evaluation~\cite{wolpert2013probability}.

In \textit{backpropagation}, discounted updates account for evolving agent intentions. Let $N_i$ be the discounted visit count for node $i$:
\begin{equation}
   N_i = \sum_{t=1}^T \gamma^{T-t} \mathbf{1}_{\{a^t = i\}}, \qquad \gamma \in [0.5,1),
\end{equation}
with the corresponding discounted value estimate
\begin{equation}
    \bar{X}_{i,N_i} = 
    \frac{1}{N_i}
    \sum_{t=1}^{T} \gamma^{T-t} r^t \mathbf{1}_{\{a^t = i\}}.
\end{equation}

\noindent where $r^t$ is the rollout score at iteration $t$ as defined in (\ref{eq: marginal}). After exhausting the computation budget, each agent selects the rollout in $\hat{\mathcal{A}}^n$ with the highest probability under $p^n$.

\subsection{Boltzmann Selection Policy}
\label{sec:boltzmann_select}

Selection in distributed MCTS is challenging due to non-stationary node statistics: deeper expansions alter reward distributions, and stochastic oversampling of low-value nodes can hinder search efficiency. To address this, CB-MCTS employs a temperature-controlled Boltzmann policy with entropy regularization and decaying uniform exploration.
For a node $i$ with children $\mathcal{C}(i)$, the probability of selecting child $j$ at iteration $t$ is
\begin{align}
    \pi_{i,t}(j)
    = (1 - \lambda_{i,t})\,\rho_{i,t}(j)
    + \frac{\lambda_{i,t}}{|\mathcal{C}(i)|},
\end{align}
where $\lambda_{i,t} = \min\!\left(1, \epsilon / \log(e+N_i)\right)$ introduces controlled uniform exploration, and
\begin{equation}
    \rho_{i,t}(j) \propto 
    \exp\!\left(
        \frac{
            \bar{X}_{j,N_j} + \beta(N_i) H_j
        }{
            \alpha(N_i)
        }
    \right)
\end{equation}
is an entropy-regularized Boltzmann distribution. Here $\alpha(\cdot)$ and $\beta(\cdot)$ are decaying schedules, and $H_j$ is the entropy bonus promoting structured early exploration. The entropy is initialized as $H_i = 0$ when the node $i$ is first expanded and dynamically backed up in the backpropagation phase as:
\begin{equation}
    H_i \leftarrow 
    \mathcal{H}(\pi_{i,t}) 
    + \sum_{j \in \mathcal{C}(i)} \pi_{i,t}(j)\, H_j,
\end{equation}
with $\mathcal{H}$ denoting Shannon entropy.

\subsection{Simple Regret Analysis}

CB-MCTS is designed to minimize simple regret by combining structured stochastic exploration with discounted coordination signals. The Boltzmann policy ensures that all actions remain discoverable while gradually concentrating probability mass on high-value branches. Meanwhile, the marginal contribution objective and discounted backups attenuate outdated information, enabling each agent to adapt to the evolving intentions of others. Together, these mechanisms promote consistent alignment between local rollouts and the global objective, allowing the search to converge more rapidly toward high-reward regions.

\begin{theorem}
The simple regret of CB-MCTS, with $\alpha(\cdot) \rightarrow 0$ and $\beta(\cdot) \rightarrow 0$, is bounded by $\mathbb{E}[r_T] \leq C \exp\!\left( - kT / \log T \right)$ for some constants $C, k > 0$.
\end{theorem}

Theorem 2 shows that the simple regret of CB-MCTS decays exponentially faster in $T$ than Dec-MCTS with D-UCT. As illustrated by the D-chain problem (Figure~\ref{fig:d_chain_main}), the simple regret of CB-MCTS vanishes with far fewer iterations, regardless of the value of $\gamma$, indicating that it identifies optimal actions more rapidly. This is especially valuable in applications with limited planning resources. The proof of Theorem 2 and additional simple regret analysis in the D-chain problem are in the Appendix.

\begin{figure}[tb]
    \begin{center}
        \includegraphics[width=0.55\linewidth]{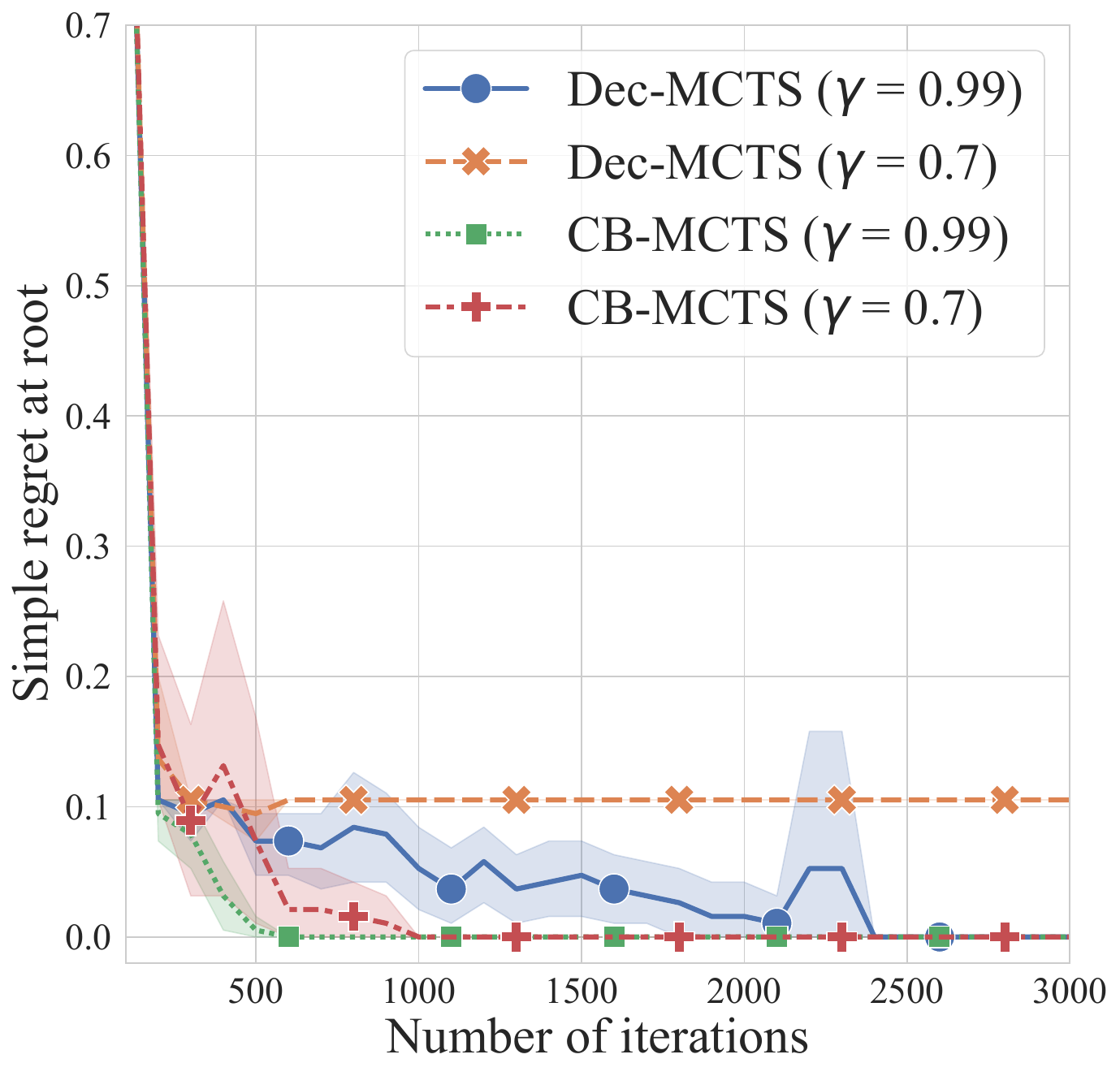}
        \caption{Simple regret of CB-MCTS and Dec-MCTS in the multi-agent D-chain problem with $D = 10$ and 2 agents.}
        \label{fig:d_chain_main}
    \end{center}
\end{figure}

\section{Empirical Evaluation}

\begin{figure*}[tb]
    \small
    \centering
    \includegraphics[width=0.95\linewidth]{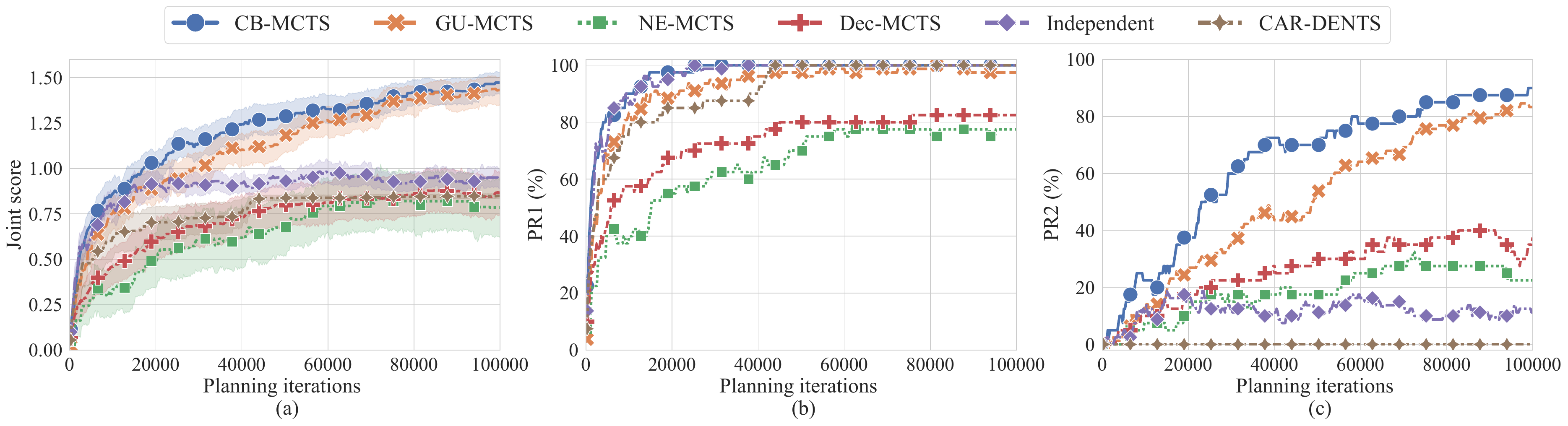}
    \caption{Performance comparison on the Frozen Lake benchmark.}
    \label{fig:frozenlake_main}
\end{figure*}

\begin{figure*}[tb]
    \small
    \centering
    \includegraphics[width=0.95\linewidth]{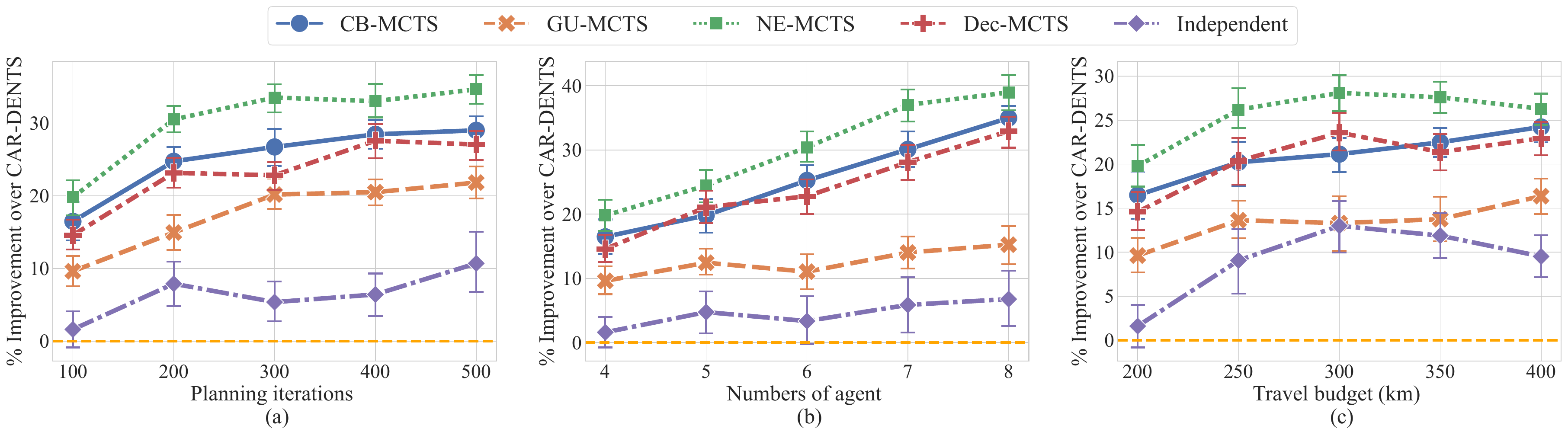}
    \caption{Performance comparison in the Oil Rigs Inspection problem.}
    \label{fig:oilrig_main}
\end{figure*}

We evaluate CB-MCTS on different multi-agent cooperative planning tasks. We consider the following baselines:

\begin{itemize}
    \item \emph{Dec-MCTS}: leading version of decentralized MCTS with D-UCT \cite{best2019dec}.
    \item \emph{GU-MCTS}: CB-MCTS using global utility instead of marginal contributions.
    \item \emph{NE-MCTS}: CB-MCTS without entropy $\left(\beta(m)=0\right)$.
    \item \emph{Independent}: CB-MCTS runs independently per agent (equivalent to the single-agent algorithm AR-DENTS \cite{painter2023monte}).
    \item \emph{CAR-DENTS}: AR-DENTS adapted for centralized multi-agent planning where a single tree encodes all agents, with agent $n$ acting at depths $(n, n+N, n+2N,\dots)$.
\end{itemize}

We tuned all methods on small validation instances. The hyperparameter selections and environment setup details are in the Appendix.

\subsection{Frozen Lake Problem}

We consider the \emph{Frozen Lake} benchmark \cite{towers2024gymnasium}, a grid-world where each cell is either safe or a hole. The agent starts in the top-left corner and moves until reaching a goal, falling into a hole, or exhausting its budget. We extend the task to a multi-agent setting with two goal positions. An agent reaching a goal at step $t$ receives a score of $0.99^t$; otherwise it receives $0$. Multiple agents selecting the same goal do not yield additional reward. This setup generalizes the multi‑goal stochastic navigation problems. Beyond the joint score, we measure the probability that at least one goal is reached (\textit{PR1}) and the probability that both goals are reached (\textit{PR2}). Each algorithm is evaluated over 80 runs on four $8{\times}12$ maps with two goals, using a planning budget of 100 steps and reporting metrics every 250 iterations.

As shown in Figure~\ref{fig:frozenlake_main}, CB-MCTS reaches both goals up to $40\%$ more often than Dec-MCTS and attains a $70\%$ higher joint score. The problem’s sparse reward structure favors Boltzmann-based exploration, which increases the chance of discovering successful trajectories. The entropy-guided search further mitigates premature termination by avoiding low-entropy (hole-adjacent) actions. Without this mechanism, NE-MCTS exhibits a substantial performance drop.

Independent and CAR-DENTS can reach at least one goal (Figure~\ref{fig:frozenlake_main}b) but frequently miscoordinate, sending both agents to the same target. GU-MCTS can eventually match CB-MCTS in joint score, but directly optimizing the global utility yields high-variance estimates and unstable coordination. In contrast, CB-MCTS leverages marginal contributions to decouple each agent’s influence, enabling faster convergence. As shown in Figure~\ref{fig:frozenlake_main}c, CB-MCTS achieves a $60\%$ PR2 level twice as fast and an $80\%$ PR2 level 1.5× faster than GU-MCTS. Overall, the results demonstrate that the components of CB-MCTS collectively provide robust performance on reward-sparse, decentralized planning tasks.

\subsection{Oil Rigs Inspection Problem}

We consider an oil rig inspection task in a 200 km $\times$ 100 km region containing 1000 oil rigs (ORs). A team of $N$ autonomous vehicles, starting from a common depot, must visit 200 randomly selected ORs, each with a 2.5 km observation radius. Agents move on an undirected graph $G$ covering the region, where every OR intersects at least one edge. An OR is \emph{visited} when an agent traverses an edge within its observation range. All agents know $G$ and plan paths to maximize total OR coverage under a uniform travel budget. This setup generalizes multi-robot informative path planning problems and is widely used for evaluating decentralized MCTS methods \cite{best2019dec, nguyen2022multi}. Each algorithm is run 40 times over 4 OR subsets. Performance is measured as the percentage of visited ORs. CAR-DENTS is evaluated in a centralized training, decentralized execution (CTDE) regime, where all agents’ actions are planned offline for 3000 iterations. Distributed algorithms use online replanning: each agent plans, executes its first edge, and replans until exhausting its budget.

Figure~\ref{fig:oilrig_main} summarizes performance under varying parameters (default: 4 agents, 100 planning iterations per cycle, 200 km travel budget). Despite the dense and smooth reward landscape, which typically favors UCT-style planners, CB-MCTS consistently matches Dec-MCTS and surpasses it with additional planning iterations. Dense rewards also increase coordination complexity, since overlapping coverage reduces global value. Accordingly, GU-MCTS (optimizing global utility directly) and Independent perform notably worse due to high-variance value estimates and limited coordination. The online distributed methods further benefit from parallelization, allowing agents to expand deeper local search trees and outperform the CTDE baseline.

Notably, NE-MCTS consistently performs best and maintains a $5$–$10\%$ improvement over Dec-MCTS. This suggests that in environments with dense, smooth reward distributions, removing entropy can lead to better empirical performance, as the Boltzmann temperature schedule effectively controls exploration. It also improves computational efficiency by lowering search variance and runtime overhead. Taken together, these results show that CB-MCTS is scalable and adaptable to a wide range of multi-agent planning problems, from smooth to sparse reward environments.

\section{Conclusion}

Efficient coordination in multi-agent planning remains challenging, especially when optimal actions initially appear suboptimal. We introduced CB-MCTS, a distributed algorithm that promotes early exploration while optimizing collective utility. Experiments show that CB-MCTS matches state-of-the-art methods in general settings and significantly outperforms them in deceptive environments requiring extensive exploration. Future work will explore how adversarial perturbations affect cooperative planning and evaluate the robustness of CB-MCTS under such conditions.

\section{Acknowledgments}

This research was supported in part by the UNITY-6G project, co-funded by the European Union’s Horizon Europe Smart Networks and Services Joint Undertaking (SNS JU) under Grant Agreement No. 101192650. Additional support was provided by the Swiss State Secretariat for Education, Research and Innovation (SERI), as well as by the FAIR GANDEEP and RADIANT projects.

\bibliography{aaai2026}

\section{Appendix}

\subsection{Technical Results}
This section discusses additional details of the theoretical results presented in the main paper. Specifically:
\begin{itemize}
    \item First, we define some important notation and assumptions about distributed MCTS with discounted empirical backpropagation;
    \item Second, we present preliminary results that will be useful for the proofs in the main paper;
    \item Third, we provide the proof of the suboptimal performance of Dec-MCTS in the D-chain problem;
    \item Fourth, we provide the proof of the simple regret of Dec-MCTS;
    \item Finally, we provide the proof of the simple regret of our proposed CB-MCTS algorithm.
\end{itemize}

\subsubsection{Notations and Assumptions}

For each agent executing CB-MCTS or Dec-MCTS, the distributed planning process can be modeled as a sequential decision-making problem, formulated as a stochastic multi-armed bandit (MAB). In the MCTS framework, each “arm” corresponds to a child node selected during tree traversal, enabling the application of MAB analysis to tree search.

A few notations are needed for the analysis. Let $I_t \in \{1, \ldots, K \}$ denote the arm pulled at round $t$, with $K$ being the number of possible arms. After selecting the arm $I_t=i$, we receive a stochastic payoff $X_{i,t} \in [0,1]$. The sequence of payoffs generate the stochastic process $\{X_{i,t}\}_t$,  $i =1,\ldots, K$ for $t \geq 1$. Let $\mu_{i,t}$ be the expected reward for arm $i$ at time $t$.

We first recall key aspects of distributed MCTS with discounted empirical backpropagation. At iteration $T$, given a discount factor $\gamma \in [0.5, 1)$, the discounted number of visit times for an arbitrary node $i$ is updated as:

\begin{equation}\label{eq:discounted_n}
   N_i = \sum\nolimits_{t=1}^T \gamma^{T-t} \mathbf{1}_{ \left\{ a^t=i \right\} },\
\end{equation}

\noindent where $\mathbf{1}_{ \left\{ a^t=i \right\} }$ is the indicator function that returns $1$ if node $i$ was selected at round $t$ and $0$ otherwise. The discounted empirical value for node $i$ is updated as:
\begin{equation}
\bar{X}_{i,N_i} = \frac{1}{N_i} \sum\nolimits_{t=1}^{T} \gamma^{T-t}X_{i,t} \ \mathbf{1}_{ \left\{ a^t=i \right\} }\ .
\end{equation}

Due to the interdependence among agents’ action sequences, the expected reward $\mu_{i,t}$ may change abruptly over time. Specifically, reward distributions remain stationary over certain intervals but shift at unknown time instants, referred to as \textit{breakpoints}. Therefore, the analysis proceeds under the four key assumptions.

\begin{assumption} (Independence)\label{as:1}
  Fix $1\leq i\leq K$. Let $\{\mathcal{F}_{i,t}\}_t$ be a filtration such that $\{{X}_{i,t}\}_t$ is $\{\mathcal{F}_{i,t}\}_t$-adapted and ${X}_{i,t}$ is conditionally independent of $\mathcal{F}_{i,t+1},\mathcal{F}_{i,t+2}, \ldots$ given $\mathcal{F}_{i,t-1}$. Further, there exists an integer $T_p$ such that for $t_i \geq T_p$ and $t < t_i, X_{i,t}$ is independent from $\mathcal{F}_{i,t}$.
\end{assumption}

Let $\Upsilon_t$ denote the number of breakpoints before time $t$, where a breakpoint is defined as the time instant where distributions of the rewards change.

\begin{assumption} (Finite Number of Changes) \label{as:2}
 The sequence $\{\Upsilon_t\}_t$ is known and bounded such that  $\lim_{t \rightarrow \infty} \Upsilon_t = \sup_t \Upsilon_t  < \infty$ and $\Upsilon_{t+1} \geq \Upsilon_{t}$.
\end{assumption}
We also assume that the expected payoff $\mu_{i,t}$ converges.
\begin{assumption} (Convergence of means) \label{as:3}
 The limit $\mu _i = \lim _{t \rightarrow \infty} \mu_{i,t}$ exists for all $i \in \{1,\ldots,K\}$.
\end{assumption}
Let the difference between the two quantities be $\delta_{i,t} = \mu_{i,t}- \mu_i$. For any arbitrary time $t$, denote the optimal arm as $t_{i^*},$ and define the optimal expected payoff by $\mu_{i^*_t,t} = \max_{i \in \{1,\ldots,K\} \mu_{i,t}}.$ The average expected payoff up to time $t$ is
$$\mu_t^* = \frac{1}{t} \sum_{u=1}^{t} \mu_{i^*_u,u}.$$

The minimum difference between the optimal reward and the instantaneous reward up to time $t$ is defined as
$$\Delta_{i,t} = \min_{u \in \{1,\ldots,t\}} \{\mu_{i^*_u,u} - \mu_{i_u,u}: i\neq i^*_u \}.$$

Let $M_i(t)$ be the number of times arm $i$ is pulled following the most recent breakpoints. The following assumption requires that the drift $\delta_{i,t}$ is proportional to $\Delta_{i,t}$ after a finite burn-in period.

\begin{assumption} (Small drifts) \label{as:4}
    There exists an index $T_0(\epsilon)$ such that for any arbitrary $\epsilon >0$ and $M_i(t) \geq T_0(\epsilon), |\delta_{i,t}| \leq \epsilon \Delta_{i,t}/2$ and $|\delta^*|\leq \epsilon  \Delta_{i,t}/2$ for all $i$.
\end{assumption}

\subsubsection{Preliminary Results}

The CB-MCTS algorithm performs two fundamental tasks: \textit{exploration} and \textit{recommendation}. During the exploration phase, at each round $t$, the algorithm selects an arm $I_t=i$ according to a sampling policy and observes the corresponding reward $X_{i, t}$. After $T$ exploration rounds, the algorithm outputs a single arm $J_T$ to be pulled in a one-shot instance. The simple regret quantifies the loss incurred when the recommended arm $J_T$ is suboptimal. More formally, the \textit{simple regret} at round $T$ equals the regret of a one-shot instance of the game for the recommended arm $J_T$:

$$r_T = \mu_{i^*_T,T} - \mu_{J_T} = \Delta_{J_T}.$$ 

Recall that CB-MCTS and Dec-MCTS recommend the action sequence $a^n$ with the highest probability from the subset $\hat{\mathcal{A}^n}$, where probabilities are updated based on observed empirical returns. Consequently, the recommended action is, with high probability, the one achieving the highest empirical reward. The expected simple regret can thus be related to the probability that a suboptimal arm attains a higher empirical mean than the optimal arm. In what follows, we connect this probability to the expected number of pulls of suboptimal arms to derive the simple regret bound.

\begin{lemma}
    \label{lem:asymptotic-regret}
    Let $r_T$ denote the simple regret after $T$ iterations, and suppose the expected number of pulls for suboptimal arm $i$ is $\mathbb{E}[N_i]$. Then, the expected simple regret decays asymptotically as
    \begin{equation}
        \mathbb{E}[r_T] = \mathcal{O}\left( \sum_{i=1}^K \Delta_{i,T} \exp\left( - c_i \Delta_{i,T}^2 \mathbb{E}[N_i] \right) \right),
    \end{equation}
    \noindent
    for some constants \( c_i > 0 \) depending on the algorithm parameters.
\end{lemma}

\begin{proof}

The expected simple regret can be related to the probability of selecting a suboptimal arm as follows:
\begin{align}
\label{eq:expected_regret}
    \mathbb{E} [r_T] = \mathbb{E} [\Delta_{J_T}] &= \sum_{i=1}^K \Delta_{i,T} \mathbb{P}\{J_T = i\}\\ \nonumber
    &\leq \sum_{i=1}^K \Delta_{i,T} \mathbb{P}\{\bar{X}_{i,T} \geq \bar{X}_{i^*,T}\}
\end{align}

\noindent where the upper bound follows the fact that to be the empirical best arm, an arm $i$ must have performed better than the best arm $i^*$. Split the interval between $\mu_i$ and $\mu^*$ at $\mu_i + \frac{\Delta_i}{2}$. With Assumptions 1 and 4 satisfied, applying Theorem 4 from \cite{garivier2011upper}:

\begin{align}
\label{eq:pooly_estimate}
    \mathbb{P}\{\bar{X}_{i,T} \geq \bar{X}_{i^*,T}\} &\leq \Pr\left[ \bar{X}_{i,T} > \mu_i +\frac{\Delta_i}{2} \right]\\ \nonumber
    &+ \Pr\left[ \bar{X}_{*,T} < \mu^* - \frac{\Delta_i}{2} \right] \\ \nonumber
    \leq &2\left\lceil\frac{\log \frac{1}{1-\gamma}}{\log (1+\eta)}\right\rceil \exp \left(- \frac{\Delta_{i,T}^2 \mathbb{E}[N_i]}{2\left(1-\frac{\eta^2}{16}\right)}\right)
\end{align}
\noindent for all positive $\eta$.

Substituting (\ref{eq:pooly_estimate}) into (\ref{eq:expected_regret}), we relate the bound on the expected simple regret to the expected number of pulls of the suboptimal arms:
\begin{equation}
\label{eq:regret_pulls}
    \mathbb{E} [r_T] \leq 2\left\lceil\frac{\log \frac{1}{1-\gamma}}{\log (1+\eta)}\right\rceil \sum_{i=1}^K \Delta_{i,T} \exp \left(- \frac{\Delta_{i,T}^2 \mathbb{E}[N_i]}{2\left(1-\frac{\eta^2}{16}\right)}\right)
\end{equation}
\end{proof}
\subsubsection{Proof of Lemma 1}

We prove the lemma using contradiction. Recall that at each depth level $d < D$ in the D-Chain problem, selecting action 1 transitions to the next decision node at level $d+1$, while choosing any other action terminates the path at a leaf node and yields an immediate reward of $(D-d)/D$. At depth $D$, action 1 yields an immediate reward of 1, whereas all other actions result in 0.

During the planning phase, at each node in the tree search, Dec-MCTS with D-UCT selects the child node $i$ with the highest D-UCT scores:

\begin{equation}
    B_{i, N_i} \coloneqq \bar{X}_{i,N_i} + \sqrt{\frac{\varepsilon \log(N_j)}{N_i}}.
\end{equation}

For the optimal reward 1 to be reached for the first time, at each depth level $d < D$, we must have the confidence bound of all suboptimal actions less than the bound of action 1. That is:

\begin{equation}
    \bar{X}_{d',N_{d'}} + \sqrt{\frac{\varepsilon \log(N_{d})}{N_{d'}}} \leq \bar{X}_{d+1,N_{d+1}} + \sqrt{\frac{\varepsilon \log(N_{d})}{N_{d+1}}},
\end{equation}

\noindent with $d'$ is the leaf node resulting from the suboptimal actions. 

Since $\bar{X}_{d', N_{d'}} = (D - d)/D$ and $\bar{X}_{d+1, N_{d+1}} \leq (D - (d+1))/D$ (e.g., because the immediate reward of 1 has not been received yet), we can deduce that:

\begin{align}
    \frac{D - d}{D} + \sqrt{\frac{\varepsilon \log(N_{d})}{N_{d'}}} &\leq \frac{D - (d + 1)}{D} + \sqrt{\frac{\varepsilon \log(N_{d})}{N_{d+1}}} \nonumber \\
    \Leftrightarrow \quad \frac{1}{D} + \sqrt{\frac{\varepsilon \log(N_{d})}{N_{d'}}} &\leq \sqrt{\frac{\varepsilon \log(N_{d})}{N_{d+1}}} \nonumber \\
    \Rightarrow \quad \frac{1}{D} &\leq \sqrt{\frac{\varepsilon \log(N_{d})}{N_{d+1}}} \nonumber \\
    \Leftrightarrow \quad N_{d} &\geq \exp\left(\frac{N_{d+1}}{\varepsilon D^2} \right).
    \label{eq:exp-bound}
\end{align}

By induction, we have the lower bound on the discounted number of times the root node is visited before the optimal reward 1 is reached for the first time as:
\begin{equation}
\label{eq: lower_bound}
    N_1 = \Omega(exp(exp(...exp(1)...)))
\end{equation}

\noindent (composition of $D-1$ exponential functions).

On the other hand, from the discounted number of times the node $i$ is visited given in (\ref{eq:discounted_n}), we can derive the upper bound on the discounted number of times the root node is visited by:
\begin{equation}
\label{eq: upper_bound}
    N_1 \leq \sum\nolimits_{t=1}^T \gamma^{T-t} = \frac{1-\gamma^T}{1-\gamma} \leq \frac{1}{1-\gamma}    
\end{equation}

Combining (\ref{eq: lower_bound}) with (\ref{eq: upper_bound}), for a large enough value of $D$ such that  $$ \Omega(exp(exp(...exp(1)...))) \geq \frac{1}{1-\gamma} $$ the two bounds are contradictory.

Thus, the algorithm cannot reach the optimal node for a large tree depth $D$. That is, Dec-MCTS with UCT will recommend a non-optimal solution after a finite planning iterations in any D-chain environment with sufficient depth $D$.

Algorithm~\ref{alg:extended_dchain} illustrates how this problem can be extended to a generalized skewed environment. In this extension, suboptimal actions lead to bounded-depth subtrees that terminate according to a condition $\text{TermRules}: S \times \mathbb{N} \rightarrow \{\text{True}, \text{False}\}$. The rewards of leaf nodes in these subtrees are determined by a value function $V: S \rightarrow [0,1)$. These extensions allow the environment to capture deceptive reward structures that naturally arise in various domains, including classic grid-world games such as Frozen Lake~\cite{towers2024gymnasium} and multi-agent applications like search and rescue~\cite{rahman2022adversar}.

\setcounter{algorithm}{1}
\begin{algorithm}[tb]
    \caption{Generalized Deceptive Tree Generation}
    \label{alg:extended_dchain}
    \textbf{Input}: tree depth $D$, branching factor $K$, termination rule $\textsc{TermRules}$, value function $V(s)$
    \begin{algorithmic}[1]
        \State Initialize root node $s_0$
        \State \textsc{BuildMainTree} $(s_0, 0)$
        
        \vspace{1mm}
        \Function{BuildMainTree}{$s$, $d$}
            \For{$a = 1$ \textbf{to} $K$}
                \State Create child $s'$ and assign to $\text{children}_s[a]$
                \If{$d = D$}
                    \State $r(s') \gets \mathbf{}{1}_{\{ a = 1 \}}$  \Comment{1 if $a=1$, else 0}
                \ElsIf{$a = 1$}
                    \State \Call{BuildMainTree}{$s'$, $d+1$}
                \Else
                    \State \Call{BuildSubtree}{$s'$, $d+1$}
                \EndIf
            \EndFor
        \EndFunction
        
        \vspace{1mm}
        \Function{BuildSubtree}{$s$, $\delta$}
            \If{$\text{TermRules}(s, \delta)$}
                \State $r(s) \gets V(s)$; \Return
            \EndIf
            \For{$a = 1$ \textbf{to} $K$}
                \State Create child $s'$ and assign to $\text{children}_s[a]$
                \State \Call{BuildSubtree}{$s'$, $\delta + 1$}
            \EndFor
        \EndFunction
    \end{algorithmic}
\end{algorithm}

\subsubsection{Proof of Theorem 1}

We note that Assumptions 1-4 are also used in \cite{best2019dec}. With these assumptions satisfied, Lemma 1 from \cite{best2019dec} provides the upper bound on the number of pulls of the suboptimal arms as:
\begin{equation}
    \mathbb{E}[N_i] \leq \mathcal{O}(\sqrt{T}\log T)
\end{equation}

Combining this result with Lemma \ref{lem:asymptotic-regret} yields the simple regret:

\begin{equation}
    \mathbb{E}[r_T] \leq  C \exp\left( - k \cdot \sqrt{T} \log T \right)
\end{equation}

\noindent with some constants $C, k > 0$.

\subsubsection{Proof of Theorem 2}
We first provide an upper bound on the number of pulls of the suboptimal arms $i$.
From the definition of $\pi_{s,t}(i)$ we decompose the expected number of pulls for arm $i$ into two parts:
\begin{align}
    \mathbb{E}[N_i] &= \sum_{t=1}^T \mathbb{E}[\pi_{s,t}(i)]\\ \nonumber
    &= \sum_{t=1}^T \left((1 - \lambda_{s,t}) \cdot \mathbb{E}[\rho_{s,t}(i)] + \frac{\lambda_{s,t}}{K} \right) \\ \nonumber
    & \leq \sum_{t=1}^T \mathbb{E}[\rho_{s,t}(i)] + \sum_{t=1}^T \frac{\lambda_{s,t}}{K}
\end{align}

Each part is the sum of pull probabilities over time due to two different factors: the first part is attributed to Boltzmann selection, and the second part is attributed to uniform exploration. We bound the second part first. Observe that $\lambda_{s,t}$ is 1 when $t$ is small and decays as $t$ grows large. Specially, there exists $t_0$ such that $\log(e+t_0) = \varepsilon$ and for $t > t_0$ we have $\lambda_{s,t} = \frac{\varepsilon}{\log(e + t)}$. The sum can then be upper-bounded by:

\begin{align}
    \sum_{t=1}^T \lambda_{s,t} &= \sum_{t=1}^T \min\left(1, \frac{\varepsilon}{\log(e + t)}\right) \\ \nonumber
    &\leq t_0 + \varepsilon\sum_{t=t_0+1}^T \frac{1}{\log(e+t)}
\end{align}

We bound the second term using an integral:
\begin{equation}
    \sum_{t=t_0+1}^T \frac{1}{\log(e+t)} \leq \int_{t_0+1}^T\frac{dt}{\log(e+t)} \leq \frac{T}{\log T}
\end{equation}

Collecting terms yields the bound on the number of pulls from uniform exploration:
\begin{equation}
    \sum_{t=1}^T \frac{\lambda_{s,t}}{K} \leq  \mathcal{O} \left( \frac{T}{K\log T} \right)
\end{equation}

We now bound the number of pulls from Boltzmann selection. Denote the number of times arm $i$ has been pulled attributed to the Boltzmann selection up until timestep $t$ as $N_{i,t}$. We split the number of pulls from Boltzmann selection into a \textit{learning phase} where each arm will be pulled at least $l_T$ times to get a decent estimated empirical reward, and a later phase where the number of pulls depends on the Boltzmann probabilities:
\begin{equation}
    \sum_{t=1}^T \mathbb{E}[\pi_{s,t}] \leq l_T +  \sum_{t=l_n+1}^T \mathbb{E}[\rho_{s,t}(i)]
\end{equation}

For simplicity, we give the exact form of $\rho_{s,t}(i)$:
\begin{align}
    \rho_{s,t}(i) &= \frac{\exp((\bar{X}_{i,t}+\beta(N_s)H_i)/\alpha(N_s))}{\sum_{j=1}^K \exp((\bar{X}_{j,t}+\beta(N_s)H_j)/\alpha(N_s))} \\ \nonumber
    &\leq \frac{\exp((\bar{X}_{i,t}+\beta(N_s)H_i)/\alpha(N_s))}{\exp((\bar{X}_{*,t}+\beta(N_s)H_*)/\alpha(N_s))} \\ \nonumber
    &= \exp \left(\frac{\bar{X}_{i,t}-\bar{X}_{*,t}+\beta(N_s)(H_i-H_*)}{\alpha(N_s)}\right)
\end{align}

Using Lemma 20 from \cite{painter2023monte}, we know that all arms will be sampled infinitely often. Thus, with Assumptions 2 and 3 satisfied, for any arm $i$, its empirical means $\bar{X}_{i,t}$ will be close to its true mean $\mu_{i,t}$ with high probability:
\begin{align}
    \bar{X}_{i,t} \in  [ \mu_{i,t} - \epsilon, \mu_{i,t} + \epsilon]\\
    \bar{X}_{*,t} \in  [ \mu_{*,t} - \epsilon, \mu_{*,t} + \epsilon]
\end{align} 

We consider the \textit{good event} where let $\epsilon = \Delta_i/4$ and with high probability:
\begin{equation}
    \bar{X}_{*,t} - \bar{X}_{i,t} \geq ( \mu_{*,t} - \epsilon) - (\mu_{i,t} + \epsilon) = \Delta_i - 2\epsilon = \frac{\Delta_i}{2}
\end{equation}

We argue that this \textit{good event} does not hold when $\bar{X}_{*,t} - \bar{X}_{i,t} \leq \Delta_i/2$, which implies either $\bar{X}_{i,t} \geq \mu_{i,t} + \Delta_i/4$ or $\bar{X}_{*,t} \leq \mu_{*,t} - \Delta_i/4$:
\begin{align}
    \mathbb{P} \left( \bar{X}_{*,t} - \bar{X}_{i,t} \leq \Delta_{i,t}/2 \right) &= \mathbb{P} \left( |\bar{X}_{i,t} - \mu_i | \geq \frac{\Delta_{i,t}}{4} \right) \\ \nonumber
    &+ \mathbb{P} \left( |\bar{X}_{*,t} - \mu_* | \geq \frac{\Delta_{i,t}}{4} \right) \\ \nonumber
\end{align}

With Assumptions 1 and 4 satisfied, applying Theorem 4 from \cite{garivier2011upper}, we can bound the first probability as:
\begin{align}
    \mathbb{P} \left( |\bar{X}_{i,t} - \mu_i | \geq \frac{\Delta_{i,t}}{4} \right) \\ \nonumber
    \leq \left\lceil\frac{\log \frac{1}{1-\gamma}}{\log (1+\eta)}\right\rceil &\exp \left(-\frac{2}{1-\frac{\eta^2}{16}}\left(\frac{\Delta_{i,t}}{4} \right)^2 N_{i,t} \right)
\end{align}

\noindent for all positive constant $\eta$. Similarly, we can also obtain a bound on the second probability. Observe that, $N_{i,t} \geq l_T$ and $N_{i,t} < N_{*,t}$, choosing $l_T =  \lceil \frac{8 \log T}{\Delta^2_{i,t}} \rceil$ we can bound the probability that the \textit{good event} does not hold as:
\begin{align}
    \mathbb{P} \left( \bar{X}_{*,t} - \bar{X}_{i,t} \leq \Delta_{i,t}/2 \right) \\ \nonumber
    \leq 2\left\lceil\frac{\log \frac{1}{1-\gamma}}{\log (1+\eta)}\right\rceil &\exp \left(- \frac{2\log T}{1-\frac{\eta^2}{16}}\right)
\end{align}

\noindent which is very small, i.e., $\mathcal{O} (1/T)$. Collecting terms yields the bound on the Boltzmann probabilities:
\begin{align}
     \sum_{t=1}^T \mathbb{E}[\rho_{s,t}(i)] &\leq \mathcal{O} (\log T) \\ \nonumber
     &+ \sum_{t=l_T+1}^T\exp \left( \frac{-\Delta_i+\beta(N_s)(H_i-H_*)}{2\alpha(N_s)} \right)
\end{align}

Since $\alpha(\cdot) \rightarrow 0$ and $\beta(\cdot) \rightarrow 0$, the second term converges to a constant. This implies that after the initial $l_T$ pulls, the total number of additional pulls from the Boltzmann component is bounded by a constant.

Combining the bounds for both parts, we get the final upper bound on the total expected pulls for arm $i$ as:
\begin{equation}
    \mathbb{E}[N_i] \leq \mathcal{O} (\log T) + \mathcal{O} \left(\frac{T}{K\log T} \right)
\end{equation}

Since $\mathcal{O} \left(\frac{T}{K\log T} \right)$ grows much faster than $\mathcal{O} (\log T)$, for sufficient large $T$, the final upper bound is:
\begin{equation}
    \mathbb{E}[N_i] \leq \mathcal{O} \left(\frac{T}{K\log T} \right)
\end{equation}

Combining this result with Lemma \ref{lem:asymptotic-regret} yields the simple regret:
\begin{equation}
    \mathbb{E}[r_T] \leq C \exp\left( - k \cdot \frac{T}{\log T} \right)
\end{equation}
\noindent with some constants $C, k > 0$.


\subsection{Additional Experimental Results}
This section discusses additional details of the experiments and results presented in the main paper. For simplicity, we set both the $\alpha(m)$ and $\beta(m)$ functions of our algorithm to $1/log(e+m)$. We note that our convergence analysis requires $\alpha(m) \to 0$ and $\beta(m) \to 0$ as $m \to \infty$ to guarantee the convergence of simple regret in theory. However, the discount factor $\gamma$ imposes an upper bound of $1/(1-\gamma)$ on the effective number of node visits, so the decaying functions $1/log(e+m)$ never fully decay to zero. However, we argue that in practice, a non-fully decay to 0 functions is still enough to promote sufficient exploration and provide good performance to the overall algorithm, given appropriate system parameter tuning. 

This section is structured as follows:

\begin{itemize}
    \item First, we provide a comprehensive analysis of the simple regret in the D-chain problem;
    \item Second, we provide details on the environments and the hyperparameter search used to select parameters in the Frozen Lake problem;
    \item Third, we provide details on the environments and the hyperparameter search used to select parameters in the Oil Rig Inspection problem;
    \item Finally, we provide a study on the gap between theoretical assumptions and practical constraints, and analyze the effect of decaying functions $\alpha(\cdot)$.
\end{itemize} 

\subsubsection{Simple Regret in D-chain Problem}

We first analyze the simple regret and compare our approach with the leading version of the decentralized MCTS with D-UCT (Dec-MCTS) \cite{best2019dec} to verify our theoretical results. To this end, we consider the \emph{multi-agent D-chain} problem (Figure \ref{fig:d_chain}). For simplicity, we set the number of actions $M$ at each depth level equal to the number of agents $N$. Each algorithm is evaluated 40 times across 4 different D-chain configurations and analyzed every 100 planning iterations under various combinations of parameter settings. For both algorithms, we considered all combinations of the following parameters (the search temperature $\alpha_{init}$ is only considered for CB-MCTS):
\begin{itemize}
    \item Exploration bias $\varepsilon$: 0.5, 1, 10, 20;
    \item Discount factor $\gamma$: 0.7, 0.9, 0.95, 0.99;
    \item Initial search temperature $\alpha_{init}$: 0.01, 0.1, 0.5, 1.
\end{itemize}

\begin{figure}[tb]
    \includegraphics[width=\linewidth]{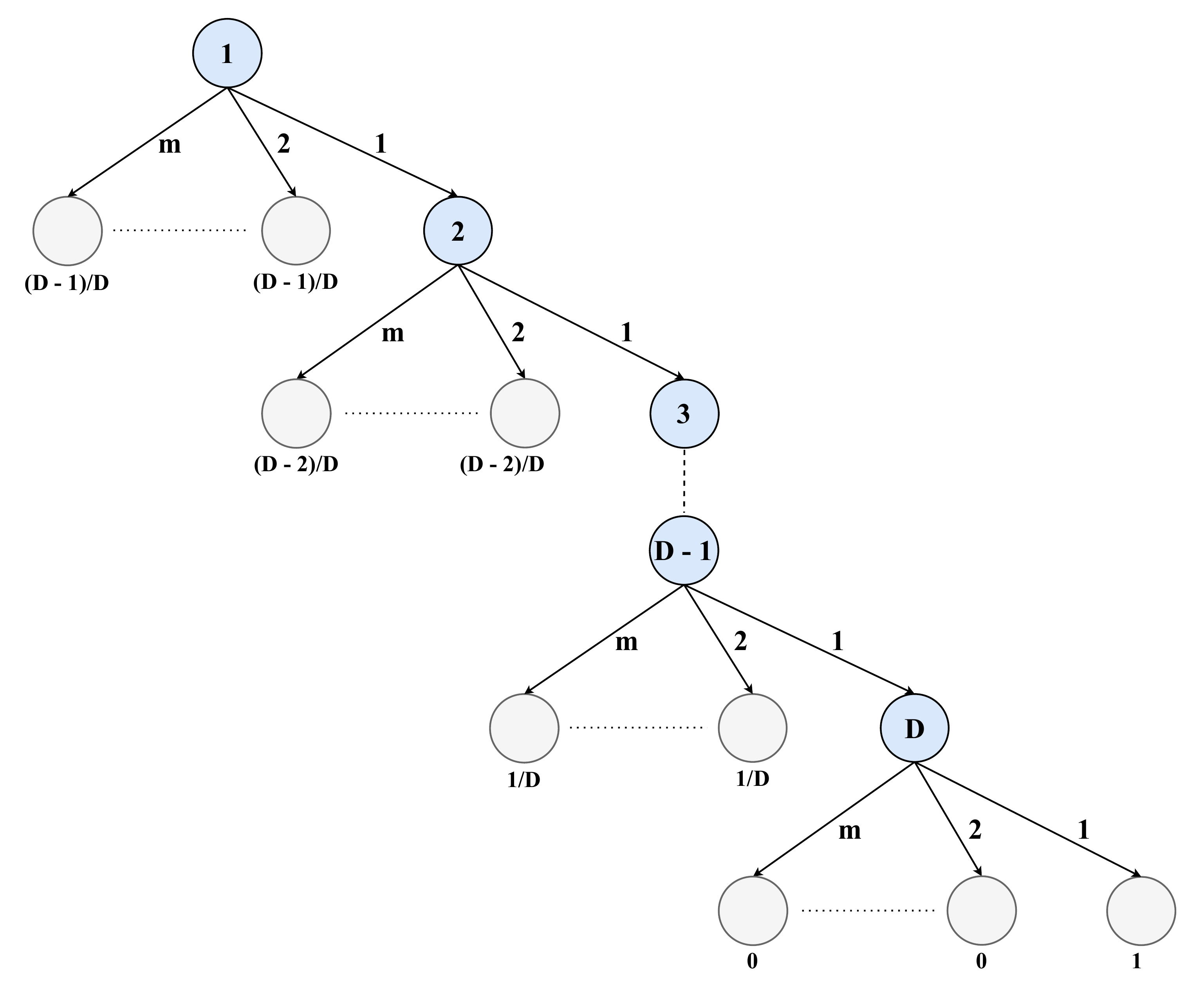}
    \caption{An illustration of the \emph{D-chain} problem on a \emph{m-ary tree} for multi-agent. Blue nodes are decision states, which can have at most $m$ children; and gray nodes are terminal states with corresponding rewards.}
    \label{fig:d_chain}
\end{figure}

Figures \ref{fig:dec_10_supp} and \ref{fig:cb_10_supp} illustrate the performance of Dec-MCTS and CB-MCTS, respectively, in a baseline scenario with two agents and a tree depth of $D=10$. While CB-MCTS consistently identifies optimal joint policies for all parameter combinations, Dec-MCTS requires a high exploration bias ($\varepsilon$) and a high discount factor ($\gamma$) to achieve similar results. This can be explained by the fact that in Dec-MCTS, $\gamma$ governs how past statistics are discounted during planning, particularly affecting the weighted mean of node utilities. A high discount factor is required for stable convergence as it ensures that older, more informative rollouts retain their influence. However, in scenarios with deceptive optimal settings, a low $\gamma$ impedes exploration and causes the algorithm to become trapped in local optima.

To further illustrate that high exploration bias and high discount factor alone are insufficient for effective search space exploration, we examined configurations with three agents and a tree depth of $D=10$, and with two agents and a tree depth of $D=20$. Figures \ref{fig:dec_33_supp} and \ref{fig:dec_20_supp} show that Dec-MCTS consistently becomes stuck in local optima in both cases. In contrast, CB-MCTS uses $\gamma$ similarly in discounted backpropagation, but its simple regret at the root node vanishes in all test cases (Figures \ref{fig:cb_33_supp} and \ref{fig:cb_20_supp}). This robustness stems from CB-MCTS's reliance on a Boltzmann search policy with a decaying entropy bonus, which naturally promotes exploration and smooths policy changes, even with more aggressive discounting. In practice, CB-MCTS can maintain strong performance across a broader range of $\gamma$ values, as it does not depend solely on reward averaging to guide search. This makes it better suited to deceptive environments, where recent rollouts are more informative and rapid adaptation is critical. However, excessively high $\varepsilon$ or $\alpha_{int}$ can cause CB-MCTS to over-explore, preventing it from following optimal tree branches to completion.

Lastly, to demonstrate how the suboptimal performance of Dec-MCTS can easily be exacerbated, we modify the D-chain problem such that at depth level $d$, the rewards for any action other than action 1 are reduced to $(D-d+1) / 2D$. Figures \ref{fig:dec_20_modified_supp} and \ref{fig:cb_20_modified_supp} present the results for Dec-MCTS and CB-MCTS in this modified scenario with two agents and a tree depth of $D=20$. The results are consistent with previous findings: only CB-MCTS successfully identifies the optimal joint policies, while Dec-MCTS remains trapped in a local optimum, even though the local optimum has been reduced by 30\%. This limitation could lead to more severe issues in practical applications, where an adversary could strategically place small rewards as decoys to deceive the agents.

\begin{figure*}[p]
	\begin{center}
        \includegraphics[width=\linewidth]{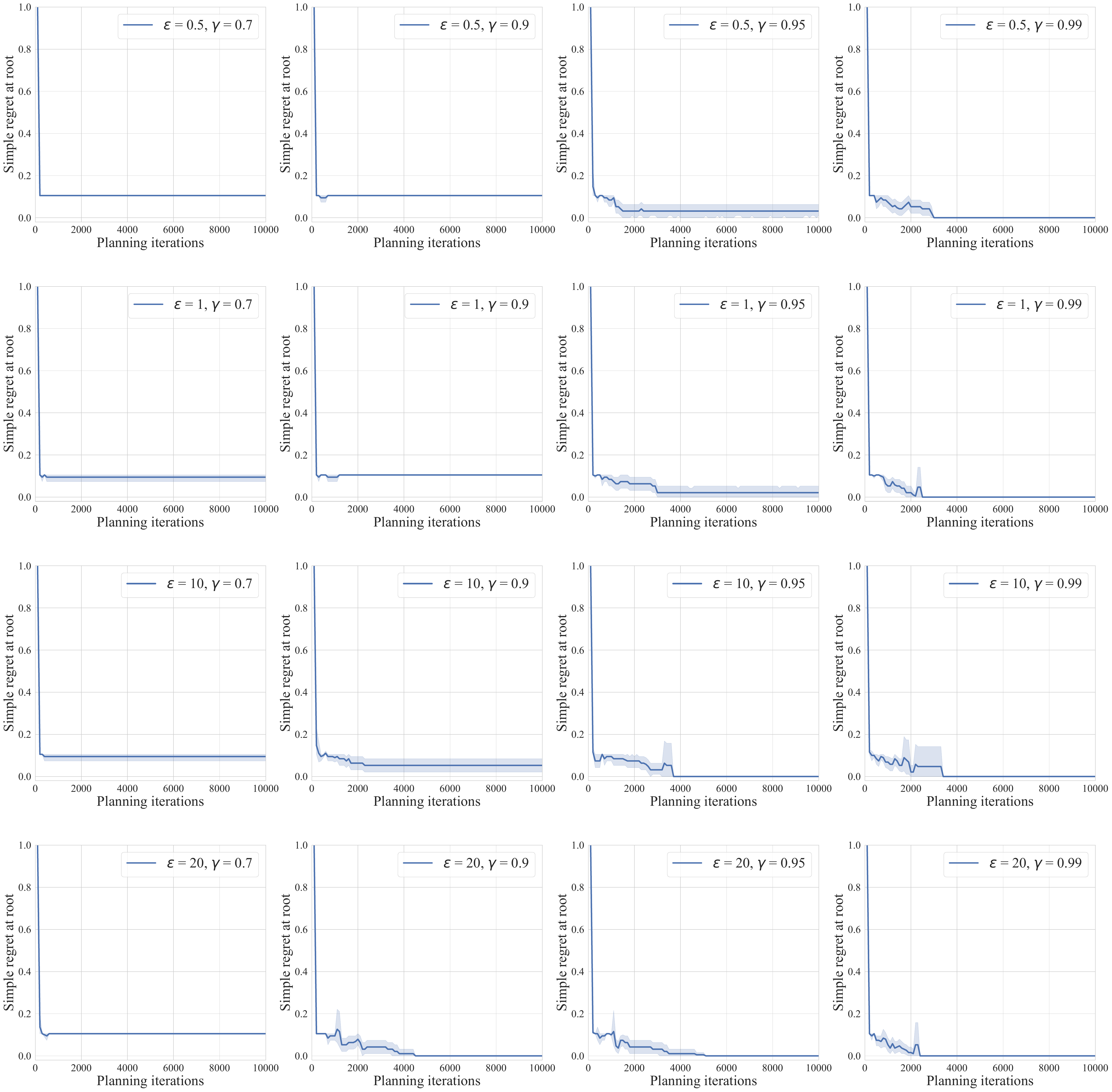}
		\caption{Results of Dec-MCTS on the D-chain problem with $D=10$ and 2 agents for varying exploration bias $\varepsilon$ and discount factor $\gamma$.
        \label{fig:dec_10_supp}\normalsize 
        }
	\end{center}
\end{figure*}

\begin{figure*}[p]
	\begin{center}
        \includegraphics[width=\linewidth]{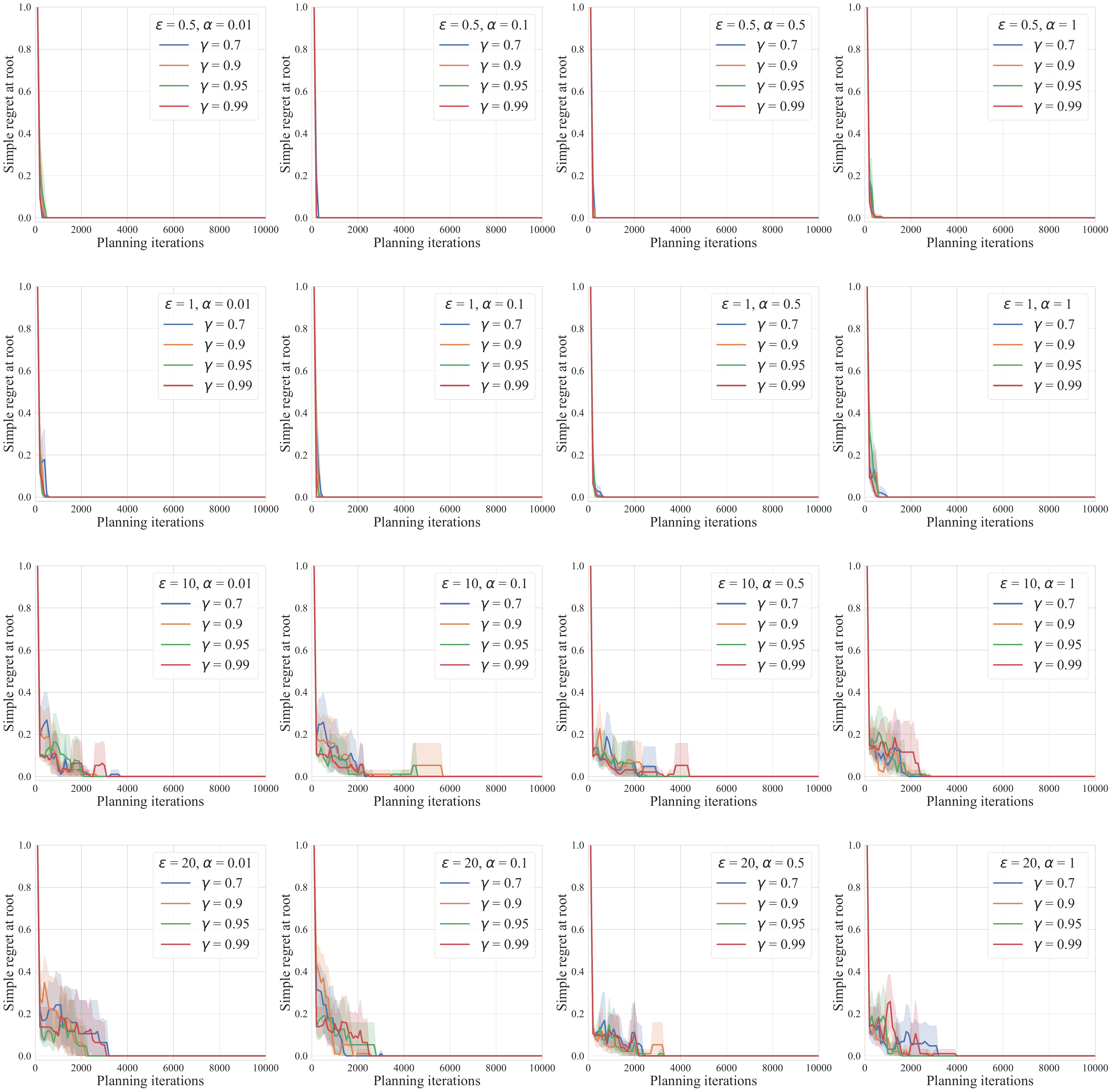}
		\caption{Results of CB-MCTS on the D-chain problem with $D=10$ and 2 agents for varying exploration bias $\varepsilon$, discount factor $\gamma$, and search temperature $\alpha$.
        \label{fig:cb_10_supp}\normalsize 
        }
	\end{center}
\end{figure*}

\begin{figure*}[p]
	\begin{center}
        \includegraphics[width=\linewidth]{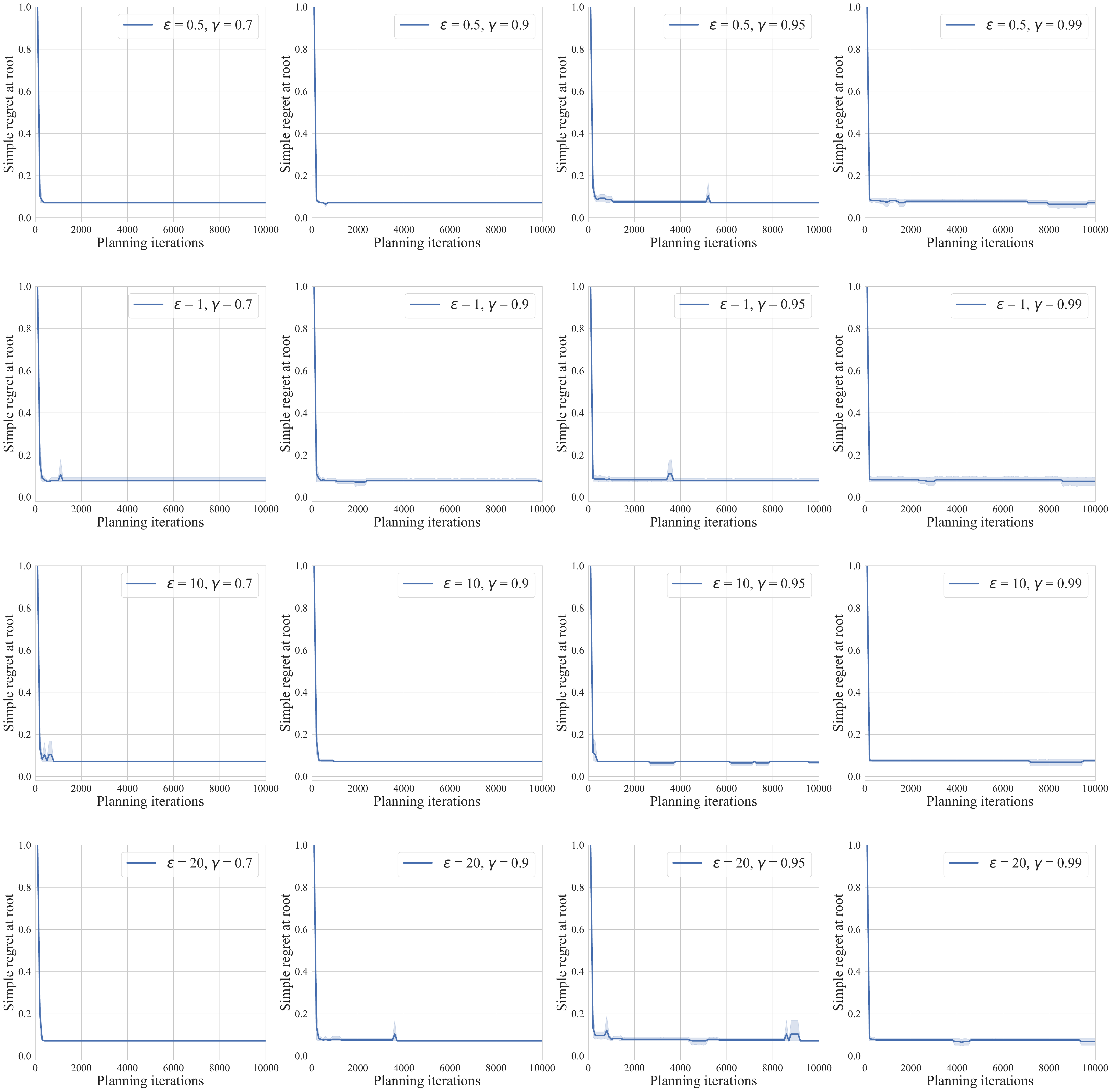}
		\caption{Results of Dec-MCTS on the D-chain problem with $D=10$ and 3 agents for varying exploration bias $\varepsilon$ and discount factor $\gamma$.
        \label{fig:dec_33_supp}\normalsize 
        }
	\end{center}
\end{figure*}

\begin{figure*}[p]
	\begin{center}
        \includegraphics[width=\linewidth]{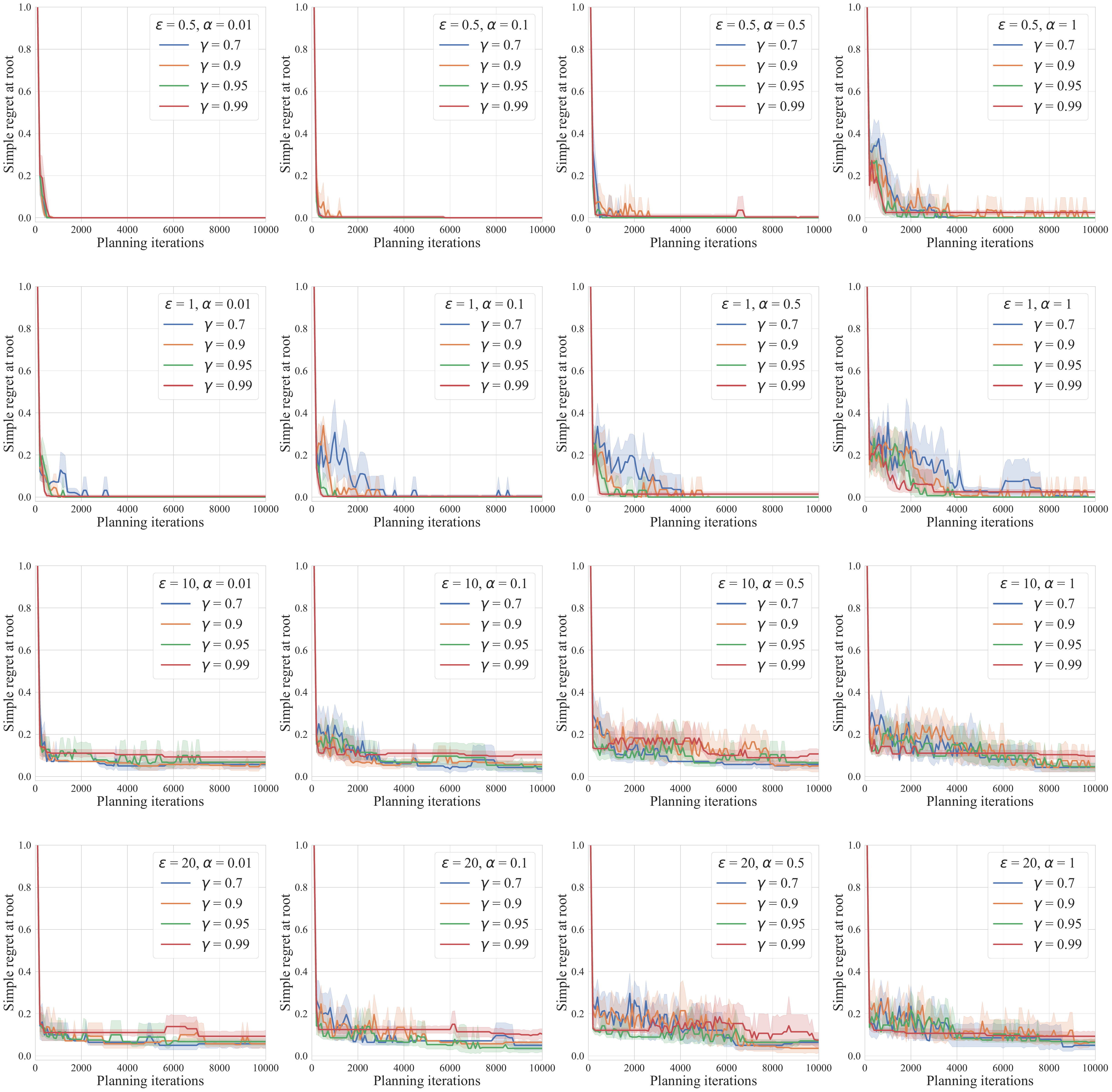}
		\caption{Results of CB-MCTS on the D-chain problem with $D=10$ and 3 agents for varying exploration bias $\varepsilon$ and discount factor $\gamma$, and search temperature $\alpha$.
        \label{fig:cb_33_supp}\normalsize 
        }
	\end{center}
\end{figure*}

\begin{figure*}[p]
	\begin{center}
        \includegraphics[width=\linewidth]{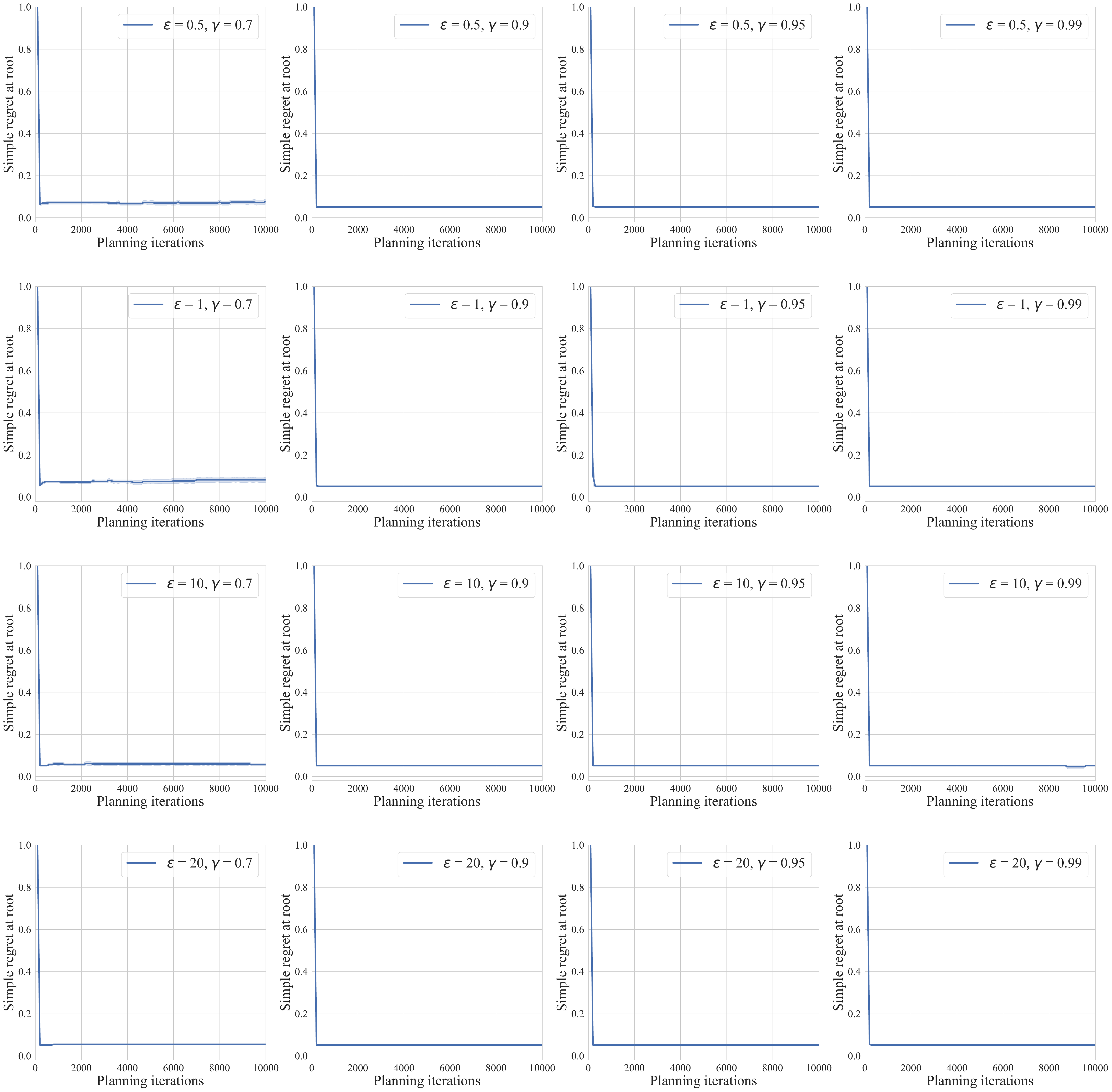}
		\caption{Results of Dec-MCTS on the D-chain problem with $D=20$ and 2 agents for varying exploration bias $\varepsilon$ and discount factor $\gamma$.
        \label{fig:dec_20_supp}\normalsize 
        }
	\end{center}
\end{figure*}

\begin{figure*}[p]
	\begin{center}
        \includegraphics[width=\linewidth]{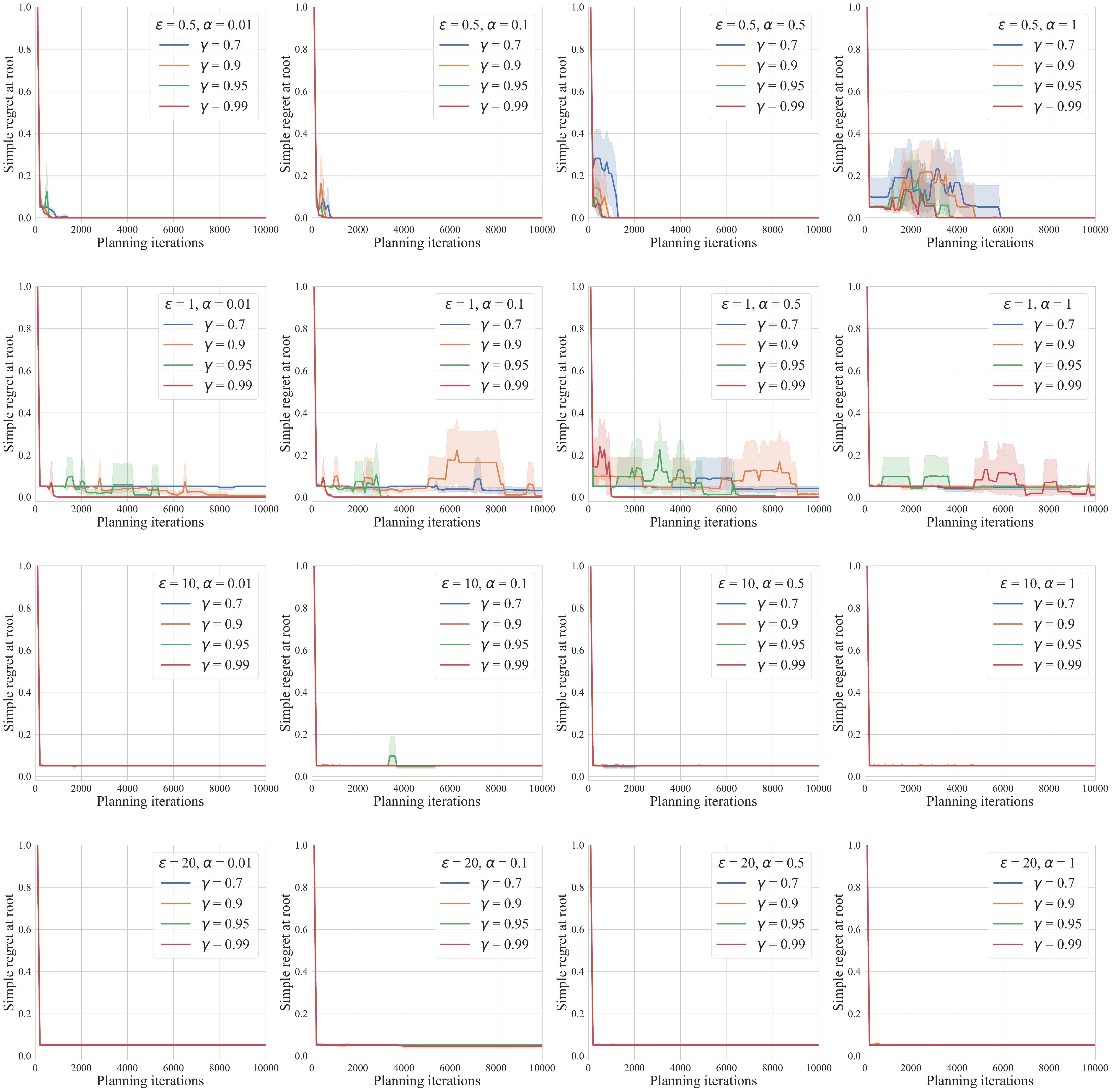}
		\caption{Results of CB-MCTS on the D-chain problem with $D=20$ and 2 agents for varying exploration bias $\varepsilon$ and discount factor $\gamma$, and search temperature $\alpha$.
        \label{fig:cb_20_supp}\normalsize 
        }
	\end{center}
\end{figure*}

\begin{figure*}[p]
	\begin{center}
        \includegraphics[width=\linewidth]{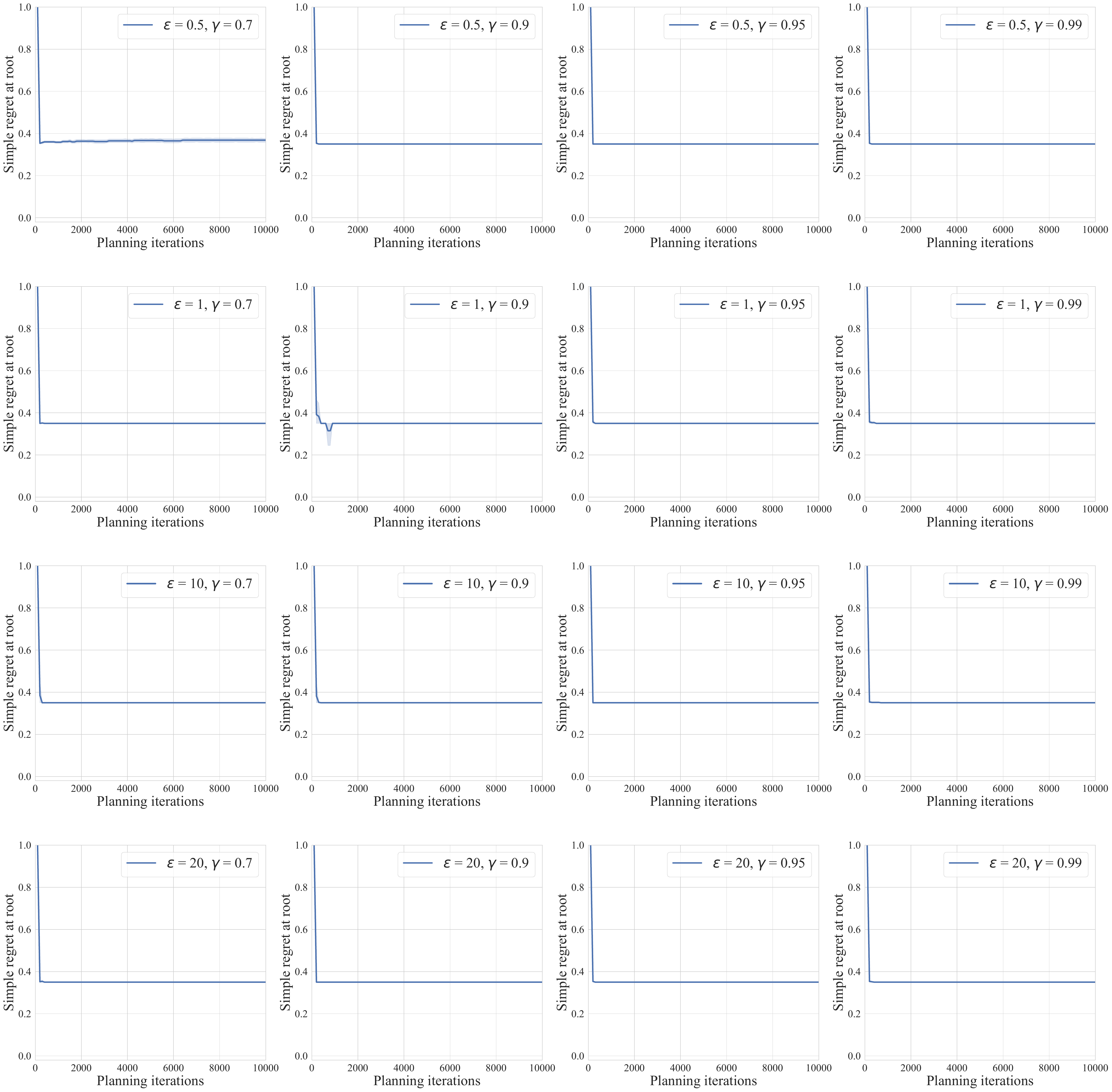}
		\caption{Results of Dec-MCTS on the \emph{modified} D-chain problem with $D=20$ and 2 agents for varying exploration bias $\varepsilon$ and discount factor $\gamma$.
        \label{fig:dec_20_modified_supp}\normalsize 
        }
	\end{center}
\end{figure*}

\begin{figure*}[p]
	\begin{center}
        \includegraphics[width=\linewidth]{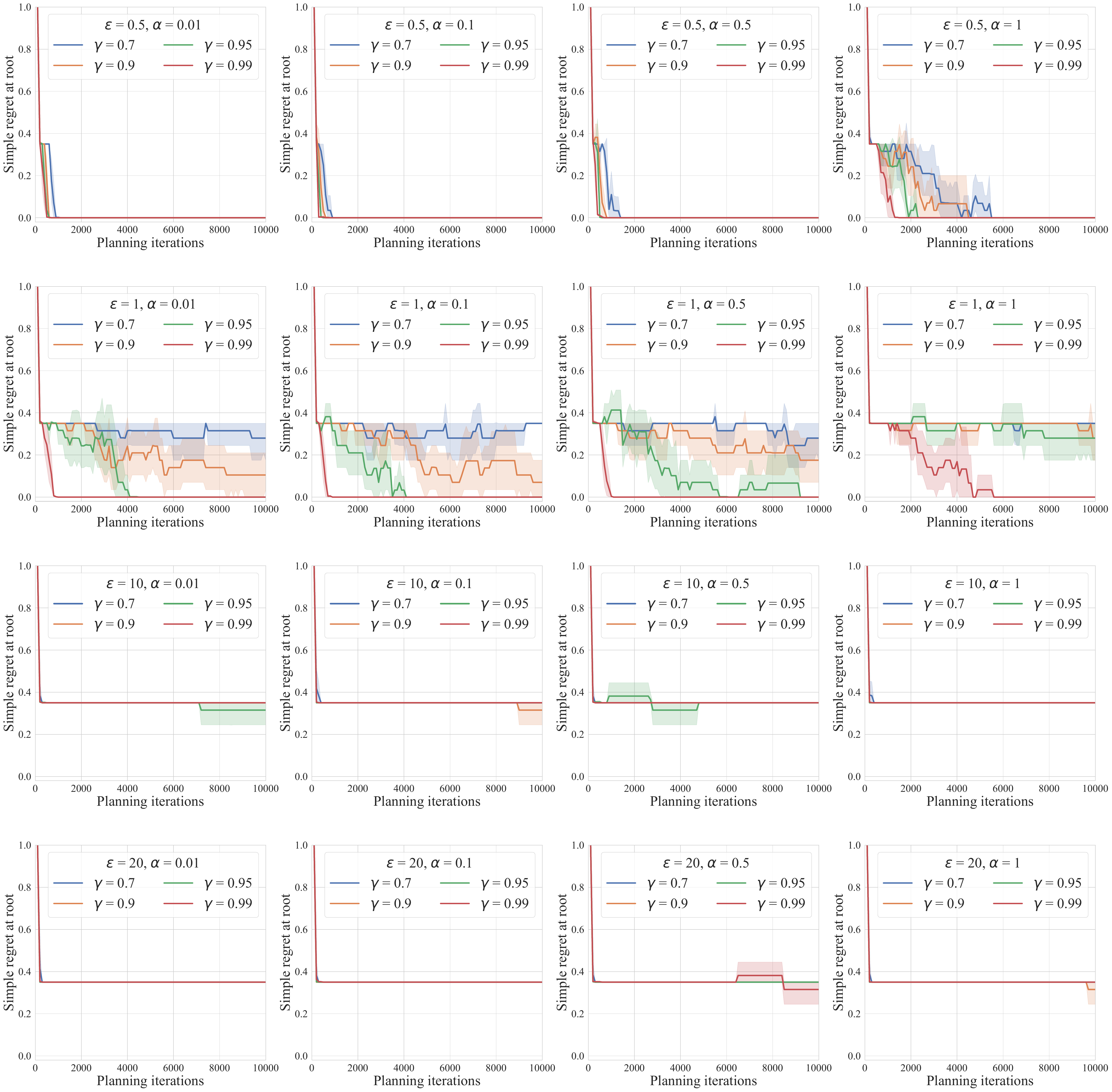}
		\caption{Results of CB-MCTS on the \emph{modified} D-chain problem with $D=20$ and 2 agents for varying exploration bias $\varepsilon$ and discount factor $\gamma$, and search temperature $\alpha$.
        \label{fig:cb_20_modified_supp}\normalsize 
        }
	\end{center}
\end{figure*}
\FloatBarrier  

\subsubsection{Frozen Lake Environment \& Parameters Details}

Figure \ref{fig:frozenlake_descs} presents the 8 $\times$ 12 Frozen Lake environments used to test the algorithms in the main paper. In each environment, \emph{red square S} is the starting location for the agents, \emph{white squares F} represent frozen plates that agents can walk on, \emph{blue squares H} are holes that end agents' trials, and \emph{yellow squares G} are the goal locations. Each environment is randomly generated, with each square having a $20\%$ probability of being a hole. All environments are verified to have a valid path to each goal using the Depth First Search algorithm.

We conduct a hyperparameter search to select the parameters for the experiments in the main paper. The search is carried out on the environment shown in Figure \ref{fig:frozenlake_descs}a, while the results presented in the main paper were averaged over the four environments. For all algorithms, we consider all combinations of the following parameters:
\begin{itemize}
    \item Exploration bias $\varepsilon$: 0.5, 1, 10, 100;
    \item Discount factor $\gamma$: 0.6, 0.7, 0.9, 0.99;
    \item Initial search temperature $\alpha_{init}$: 0.01, 0.1, 1, 10.
\end{itemize}

Table \ref{tab:frozenlake} provides the selected combinations for the experiments in the main paper, which achieved the best cumulative joint score performances. To better assess the contribution of each component in CB-MCTS, we evaluated the three ablated variants (CB-MCTS-Global, NE-MCTS, and Independent) with parameter settings identical to those of the full CB-MCTS algorithm.

\begin{table}[h]
\normalsize
\caption{Parameters used for the experiments in the main paper.}
\label{tab:frozenlake}
\centering
\begin{tabular}{lccc@{}}
\toprule 
\textbf{Algorithm}      & $\varepsilon$     & $\gamma$    & $\alpha_{init}$ \\
\midrule
Dec-MCTS                & 100       & 0.99        & - \\
CB-MCTS                 & 0.5       & 0.9         & 1           \\
CAR-DENTS               & 0.5       & -           & 1   \\
\bottomrule
\end{tabular}
\end{table}

\begin{figure}
    \begin{center}
        \includegraphics[width=0.75\linewidth]{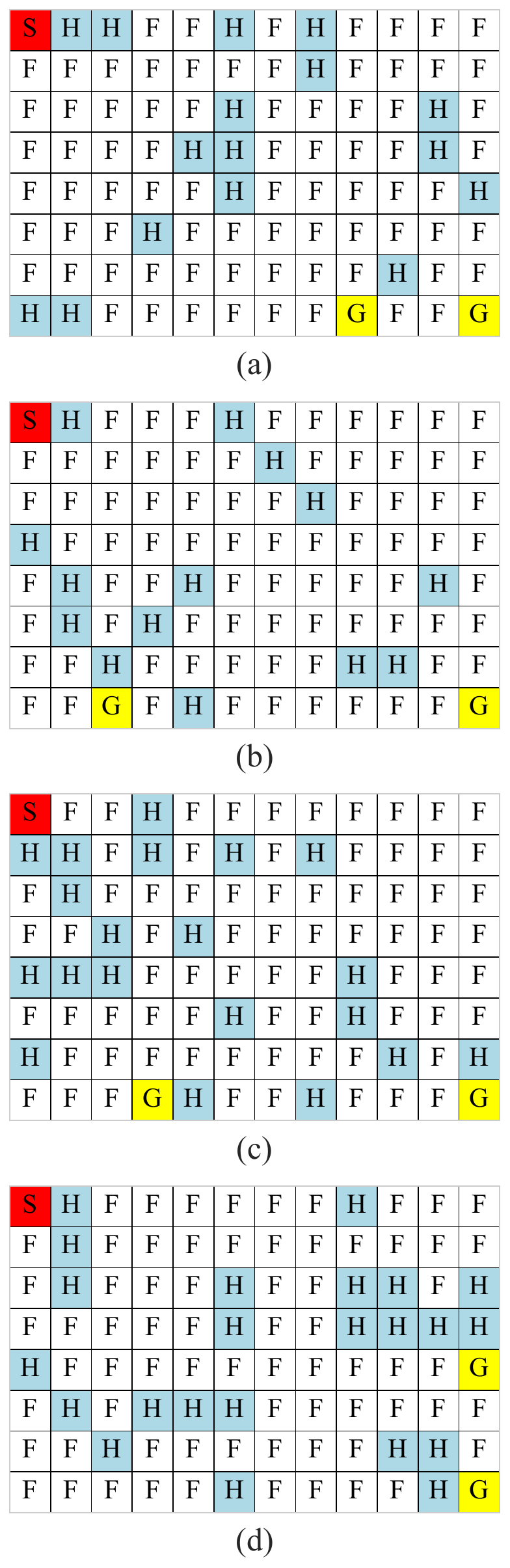}
        \caption{Environment maps used for the Frozen Lake experiments.}
        \label{fig:frozenlake_descs}
    \end{center}
\end{figure}

\begin{figure*}[t]
	\begin{center}
        \includegraphics[width=0.8\linewidth]{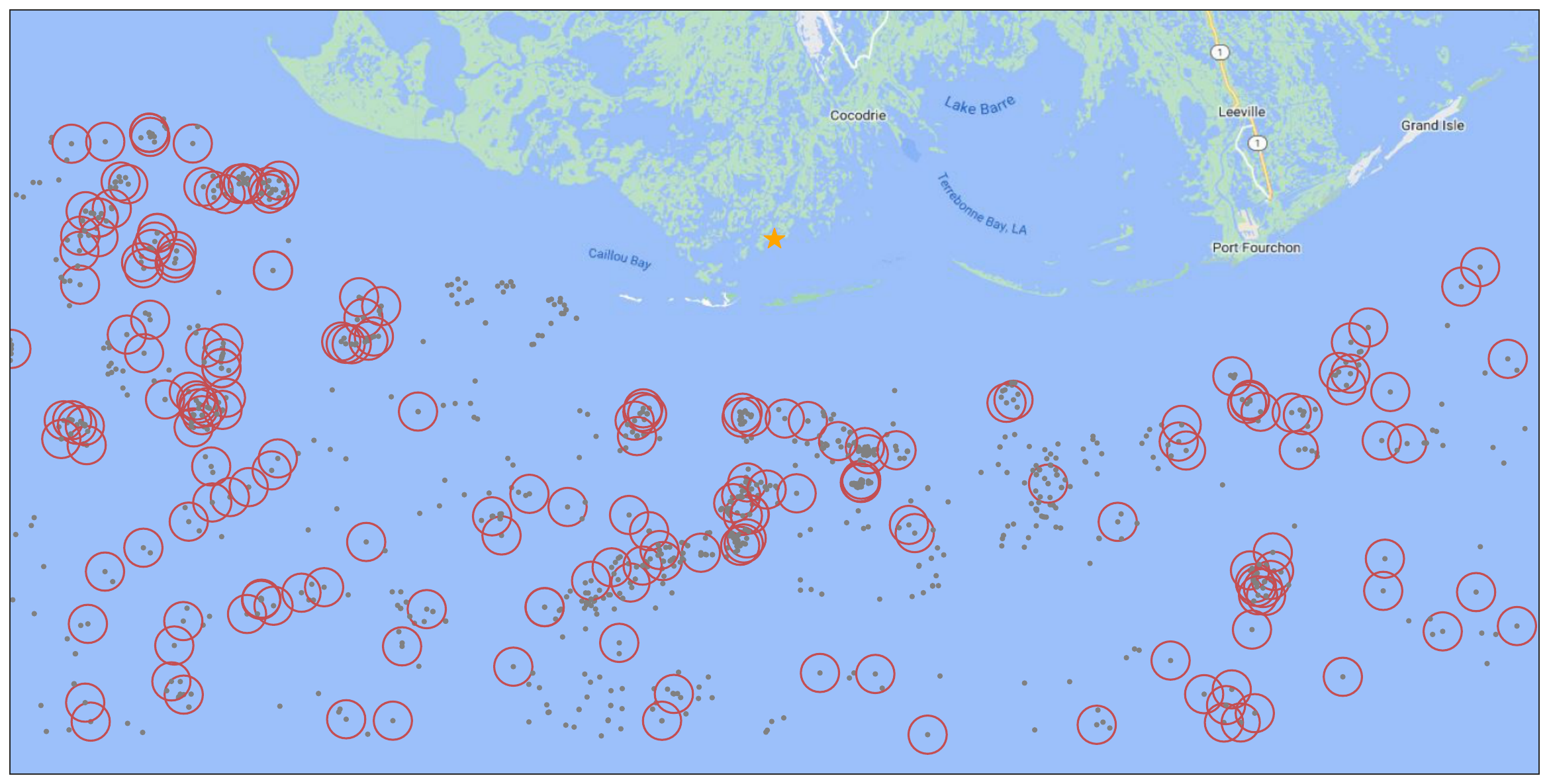}
		\caption{Environment used for parameter search for the Oil Rigs Inspection problem. \emph{Yellow star} is the starting location of all agents, \emph{grey dots} are the locations of the oil rigs, and \emph{red circles} are observation ranges of the oil rigs requiring inspection.
        \label{fig:oilrig_0_supp}\normalsize 
        }
	\end{center}
\end{figure*}
\FloatBarrier
\subsubsection{Oil Rig Environment \& Parameters Details}

The motion graph for the agents in the main paper is constructed using a probabilistic roadmap with a Dubins path model \cite{kavraki1996probabilistic}. This model employs curves to refine the straight-line segments connecting waypoints and is extensively utilized for representing motion constraints pertinent to vehicle-like nonholonomic robots such as Autonomous Underwater Vehicles \cite{41470}. The graph has 2000 vertices and 81626 edges. The ground truth locations of the 1000 oil rigs are obtained from the Bureau of Ocean Energy Management\footnote{Available at: https://www.data.boem.gov.}. Four subsets of 200 oil rigs were randomly selected for the experiments in the main paper.

We perform a hyperparameter search to select the optimal parameters for the experiments presented in the main paper. The search is performed on the environment shown in Figure \ref{fig:oilrig_0_supp} with 3 agents, a travel budget of 200 km, and a planning time of 100 iterations. The results presented in the main paper are averaged over all four different environments. For all algorithms, we consider all combinations of the following parameters:

\begin{itemize}
    \item Exploration bias $\varepsilon$: 0.5, 1, 10, 100;
    \item Discount factor $\gamma$: 0.6, 0.7, 0.8, 0.9;
    \item Initial search temperature $\alpha_{init}$: 0.01, 0.1, 1, 10.
\end{itemize}

Table \ref{tab: oilrig} gives the combinations chosen for the experiments in the main paper, which have the best median joint score performances. To better assess the contribution of each component in CB-MCTS, we evaluated the three ablated variants (CB-MCTS-Global, NE-MCTS, and Independent) with parameter settings identical to those of the full CB-MCTS algorithm.

\begin{figure}[ht]
	\begin{center}
        \includegraphics[width=0.85\linewidth]{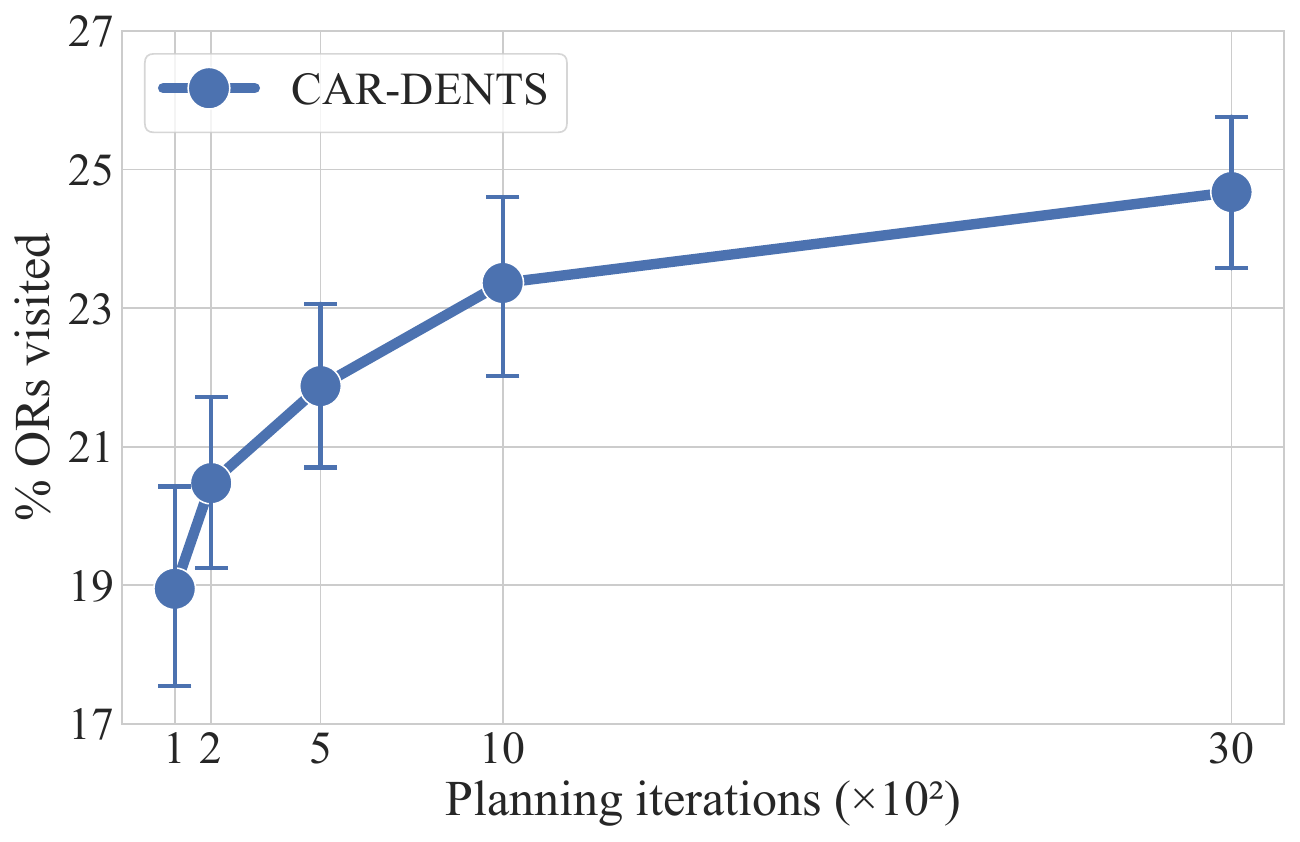}
		\caption{Performance of CAR-DENTS across different numbers of planning iterations. Results are with a $95\%$ confidence interval.
        \label{fig:oilrig_car_planning}\normalsize 
        }
	\end{center}
\end{figure}

\begin{table}[!ht]
    \normalsize
    \caption{Parameters used for the experiments in the main paper.}
    \label{tab: oilrig}
    \centering
    \begin{tabular}{lccc@{}}
        \toprule 
        \textbf{Algorithm}      & $\varepsilon$     & $\gamma$    & $\alpha_{init}$ \\
        \midrule
        Dec-MCTS                & 100       & 0.6         & - \\
        CB-MCTS                 & 0.5       & 0.8         & 0.01           \\
        CAR-DENTS               & 0.5       & -           & 10           \\
        \bottomrule
    \end{tabular}
\end{table}

After parameter tuning, we evaluate the CAR-DENTS algorithm across four environments using four agents, a 200 km travel budget, and varying numbers of planning iterations. As shown in Figure \ref{fig:oilrig_car_planning}, increasing the number of iterations initially improves performance but eventually yields diminishing returns. We therefore select 3000 iterations for the experiments in the main paper to ensure strong performance while maintaining runtime practicality.

\clearpage

\subsubsection{Search Temperature Decaying Rate Analysis}

In this section, we examine the effect of the decaying function $\alpha(m)$ on the performance of our proposed CB-MCTS algorithm to understand how the rate of search temperature decay influences exploration. To this end, we design an ablated variant \textbf{CB-MCTS with Fast-Decaying Alpha (FA-MCTS)}, which sets the temperature decay function to

\begin{equation}
    \alpha(m) = e^{-m / \left( \frac{1}{1-\gamma} - m \right)}, \nonumber
\end{equation}

\noindent while keeping the original entropy coefficient

\begin{equation}
    \beta(m) = 1 / \log(e + m). \nonumber    
\end{equation}

This formulation ensures that $\alpha(m)$ rapidly approaches zero as $m$ approaches the upper bound $1/(1 - \gamma)$, corresponding to the maximum effective discounted number of visits to a node.

The motivation for this design arises from the gap between theoretical assumptions and practical constraints. Our convergence analysis requires $\alpha(m) \to 0$ as $m \to \infty$ to guarantee sufficient exploitation in the limit. In practice, however, the discount factor $\gamma$ imposes an upper bound of $1/(1-\gamma)$ on the effective number of node visits, so the original $\alpha(m)$ function never fully decays to zero. FA-MCTS addresses this by enforcing faster temperature decay within the bounded planning horizon, allowing us to investigate the trade-offs between exploration persistence, convergence behavior, and policy robustness.

\noindent\textbf{Simple Regret in D-chain Problem} We first analyze the simple regret at the root node using the D-Chain problem. In this section, we also include the ablated version with no entropy NE-MCTS. We evaluate both FA-MCTS and NE-MCTS across four different multi-agent configurations of the D-chain problem and under the same parameter combinations as in the main experiments.

Our results show that NE-MCTS, which removes entropy regularization, is still capable of discovering joint optimal policies in the baseline setting with 2 agents and tree depth $D = 10$, particularly with a sufficiently high exploration bias $\varepsilon$ or search temperature $\alpha$ (Figure~\ref{fig:noentropy_10_supp}). However, its performance worsens in more challenging settings with greater depth or more agents (Figures~\ref{fig:noentropy_33_supp}, \ref{fig:noentropy_20_supp}, \ref{fig:noentropy_20_modified_supp}). Without entropy to maintain stochasticity, the algorithm fails to sufficiently explore the full chain, often missing the globally optimal path.

In contrast, FA-MCTS, which enforces a faster decay of the search temperature $\alpha(m)$, maintains strong performance and successfully recovers the optimal joint strategy in nearly all evaluated scenarios (Figures~\ref{fig:decaying_10_supp}, \ref{fig:decaying_33_supp}, \ref{fig:decaying_20_supp}, \ref{fig:decaying_20_modified_supp}). Nevertheless, we observe that excessively high values of $\varepsilon$ or initial temperature $\alpha_{init}$ can lead to over-exploration early in the search, especially when paired with fast decay, causing the algorithm to fail to consistently commit to optimal branches.

These findings highlight two key insights. First, entropy regularization is crucial for sustained exploration in deceptive or skewed-reward environments, enabling the agent to reach and recognize optimal rewards. Second, the algorithm remains robust to variations in the temperature decay rate, provided that the initial search temperature and exploration bias are appropriately calibrated.

\begin{figure*}[p]
	\begin{center}
        \includegraphics[width=\linewidth]{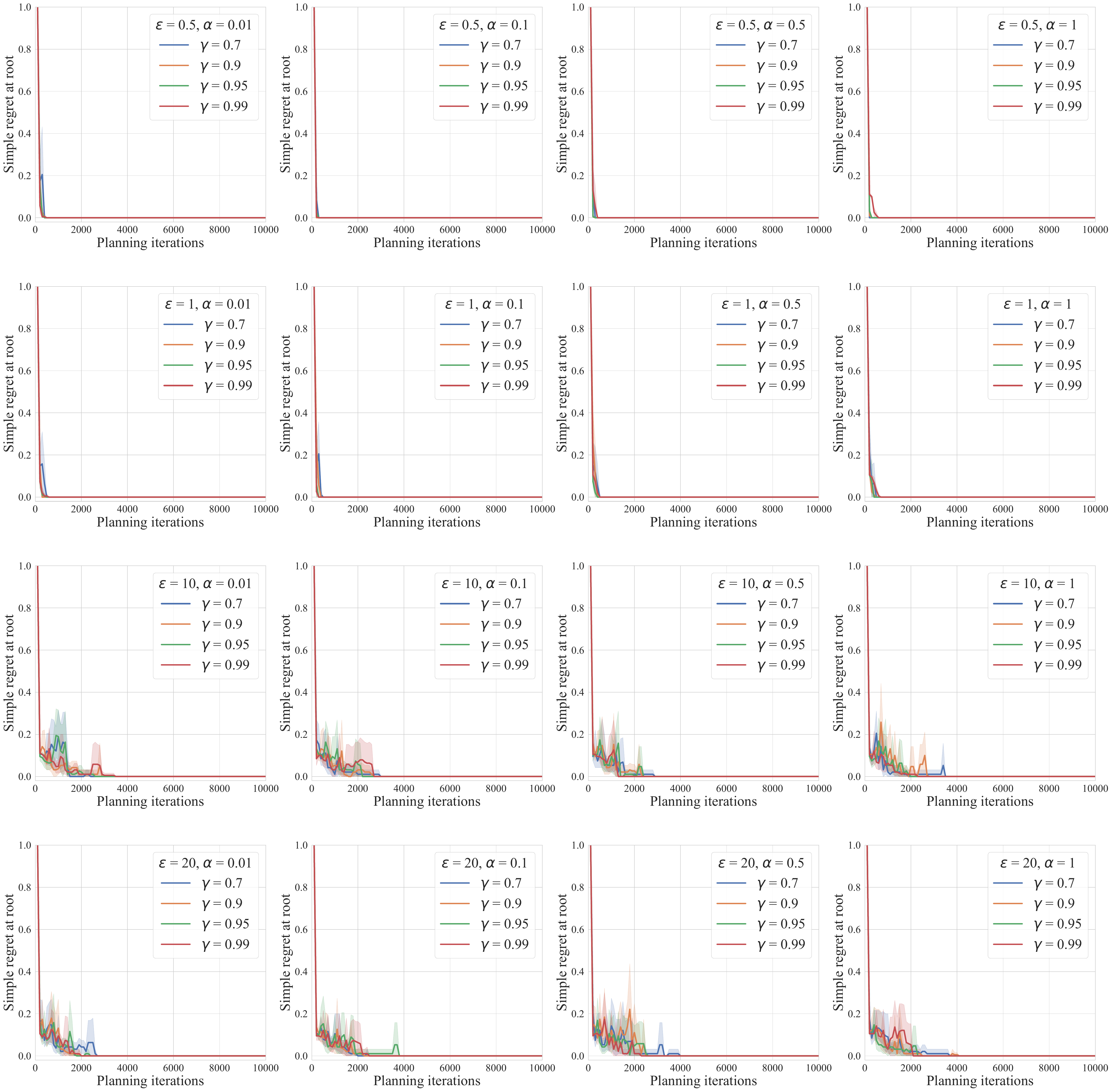}
		\caption{Results of FA-MCTS on the D-chain problem with $D=10$ and 2 agents for varying exploration bias $\varepsilon$ and discount factor $\gamma$.
        \label{fig:decaying_10_supp}\normalsize 
        }
	\end{center}
\end{figure*}
\clearpage
\begin{figure*}[p]
	\begin{center}
        \includegraphics[width=\linewidth]{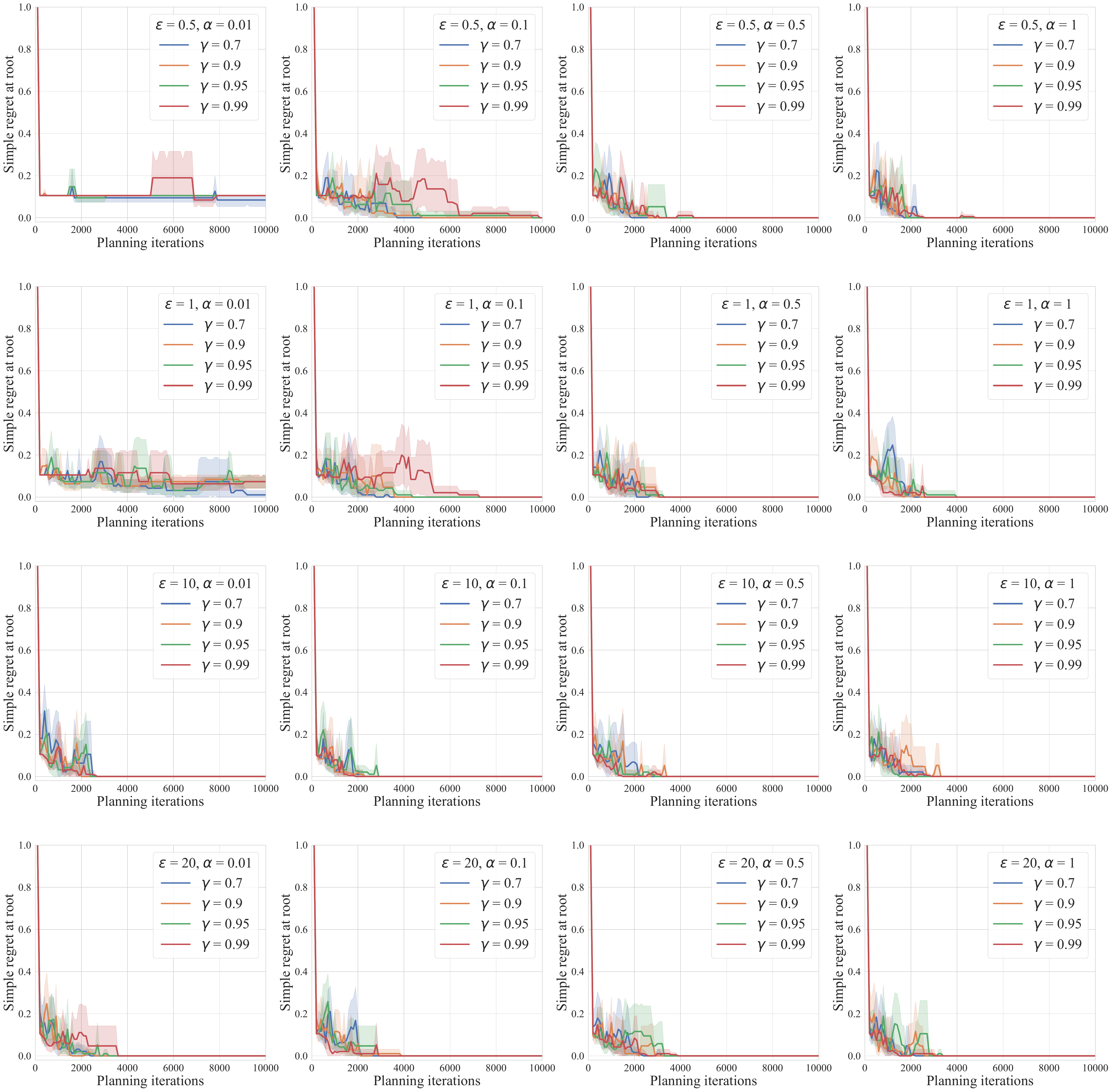}
		\caption{Results of NE-MCTS on the D-chain problem with $D=10$ and 2 agents for varying exploration bias $\varepsilon$, discount factor $\gamma$, and search temperature $\alpha$.
        \label{fig:noentropy_10_supp}\normalsize 
        }
	\end{center}
\end{figure*}
\clearpage
\begin{figure*}[p]
	\begin{center}
        \includegraphics[width=\linewidth]{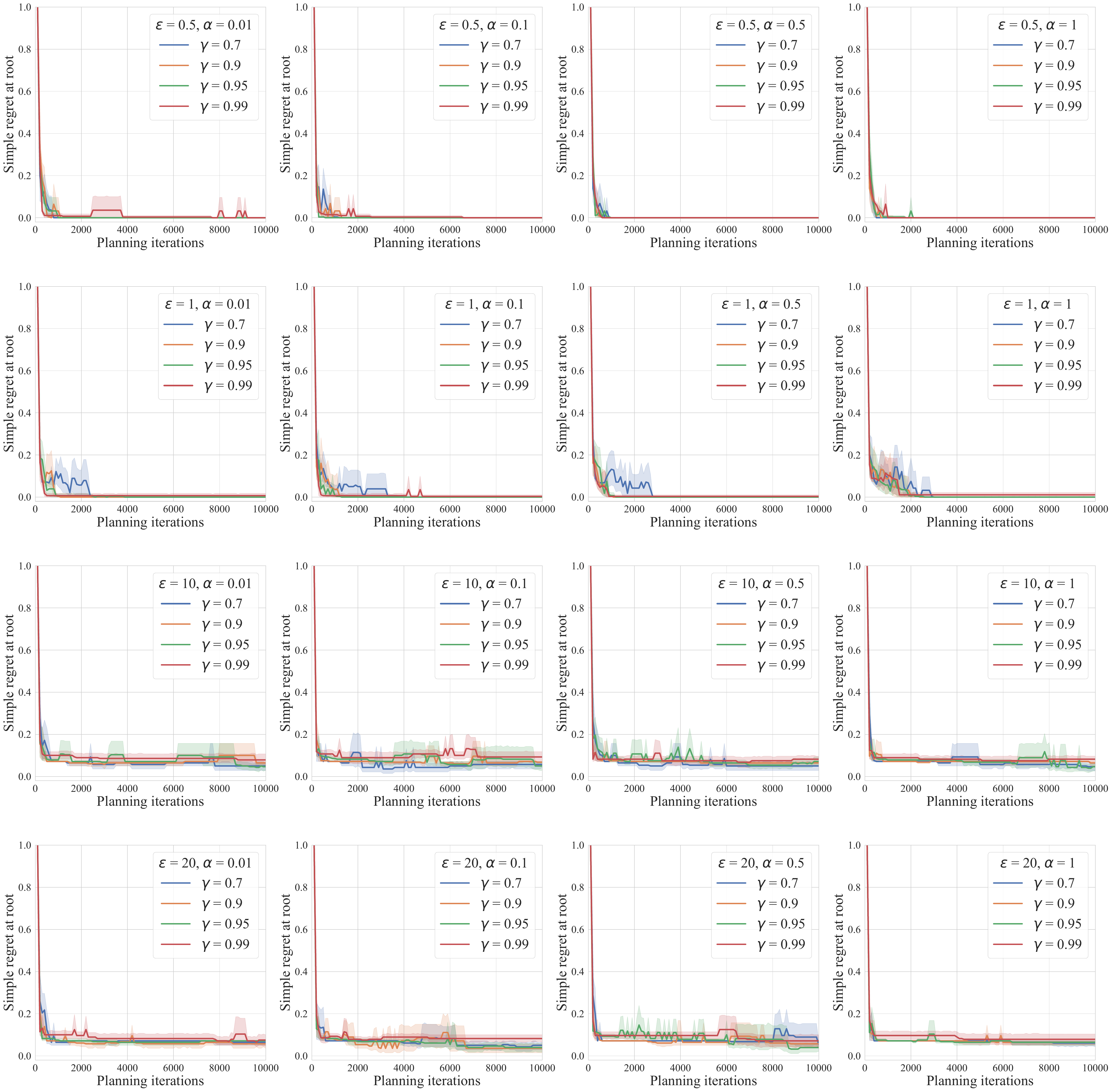}
		\caption{Results of FA-MCTS on the D-chain problem with $D=10$ and 3 agents for varying exploration bias $\varepsilon$ and discount factor $\gamma$.
        \label{fig:decaying_33_supp}\normalsize 
        }
	\end{center}
\end{figure*}
\clearpage
\begin{figure*}[p]
	\begin{center}
        \includegraphics[width=\linewidth]{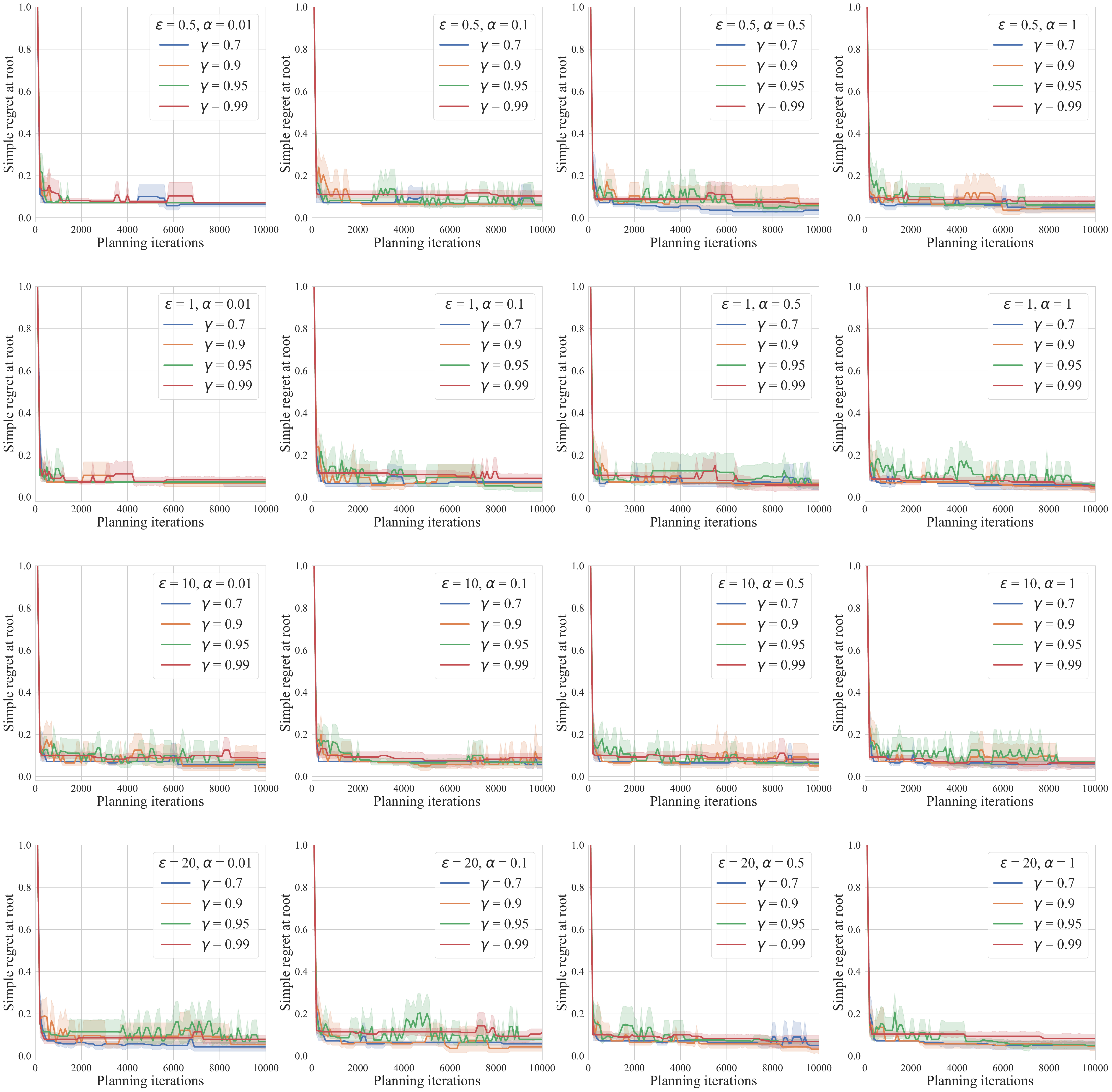}
		\caption{Results of NE-MCTS on the D-chain problem with $D=10$ and 3 agents for varying exploration bias $\varepsilon$ and discount factor $\gamma$, and search temperature $\alpha$.
        \label{fig:noentropy_33_supp}\normalsize 
        }
	\end{center}
\end{figure*}
\clearpage
\begin{figure*}[p]
	\begin{center}
        \includegraphics[width=\linewidth]{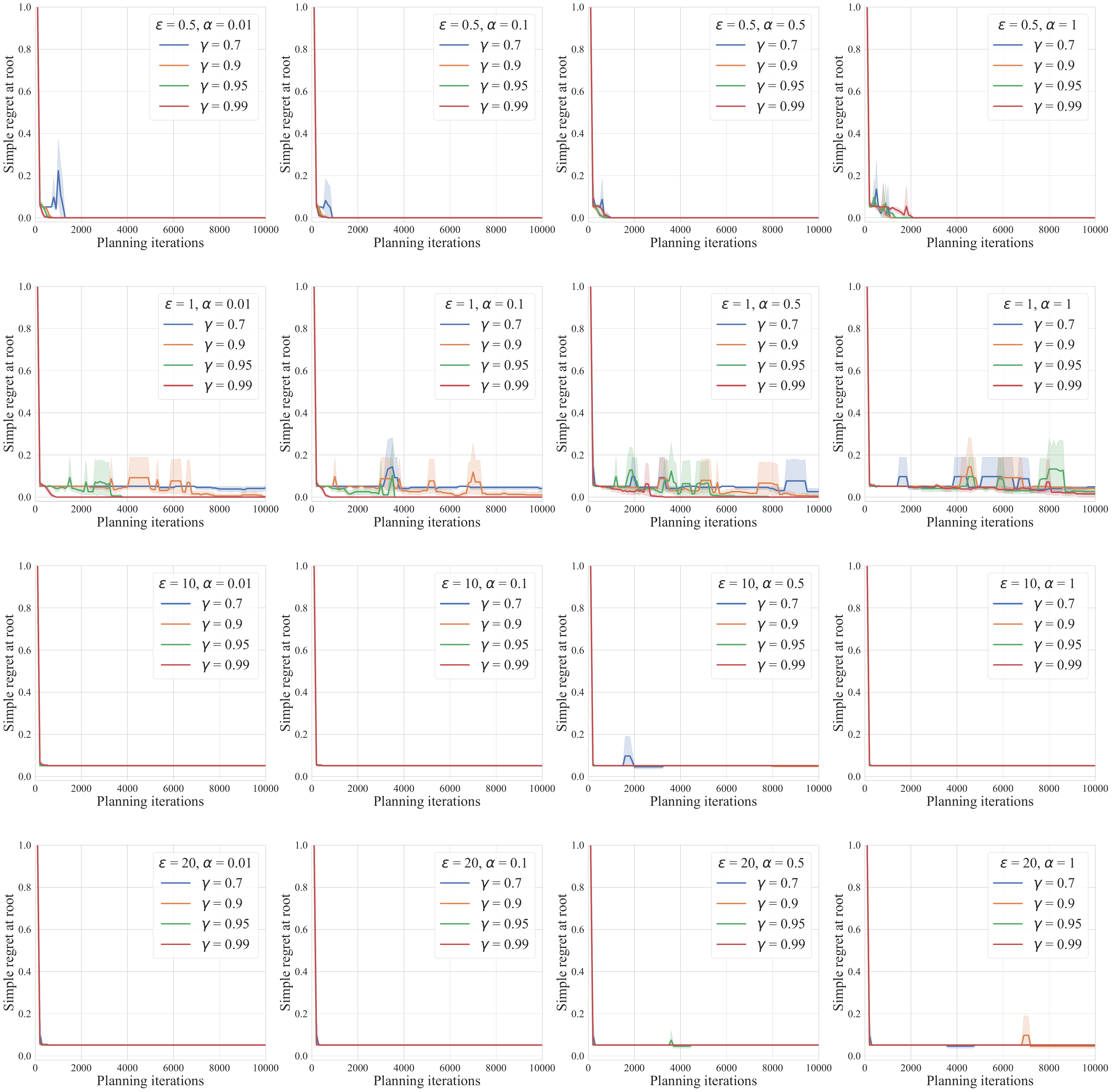}
		\caption{Results of FA-MCTS on the D-chain problem with $D=20$ and 2 agents for varying exploration bias $\varepsilon$ and discount factor $\gamma$.
        \label{fig:decaying_20_supp}\normalsize 
        }
	\end{center}
\end{figure*}
\clearpage
\begin{figure*}[p]
	\begin{center}
        \includegraphics[width=\linewidth]{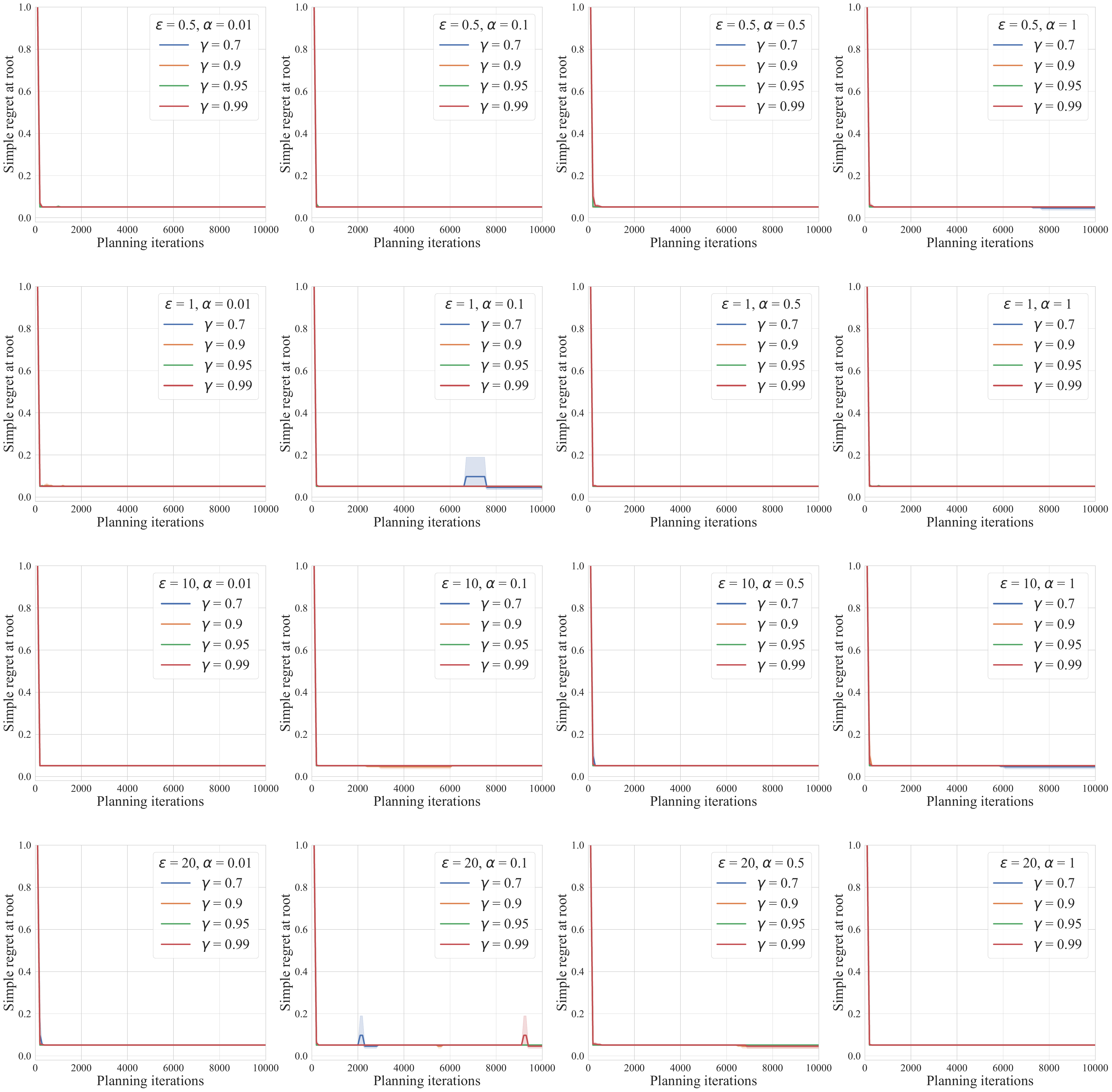}
		\caption{Results of NE-MCTS on the D-chain problem with $D=20$ and 2 agents for varying exploration bias $\varepsilon$ and discount factor $\gamma$, and search temperature $\alpha$.
        \label{fig:noentropy_20_supp}\normalsize 
        }
	\end{center}
\end{figure*}
\clearpage
\begin{figure*}[p]
	\begin{center}
        \includegraphics[width=\linewidth]{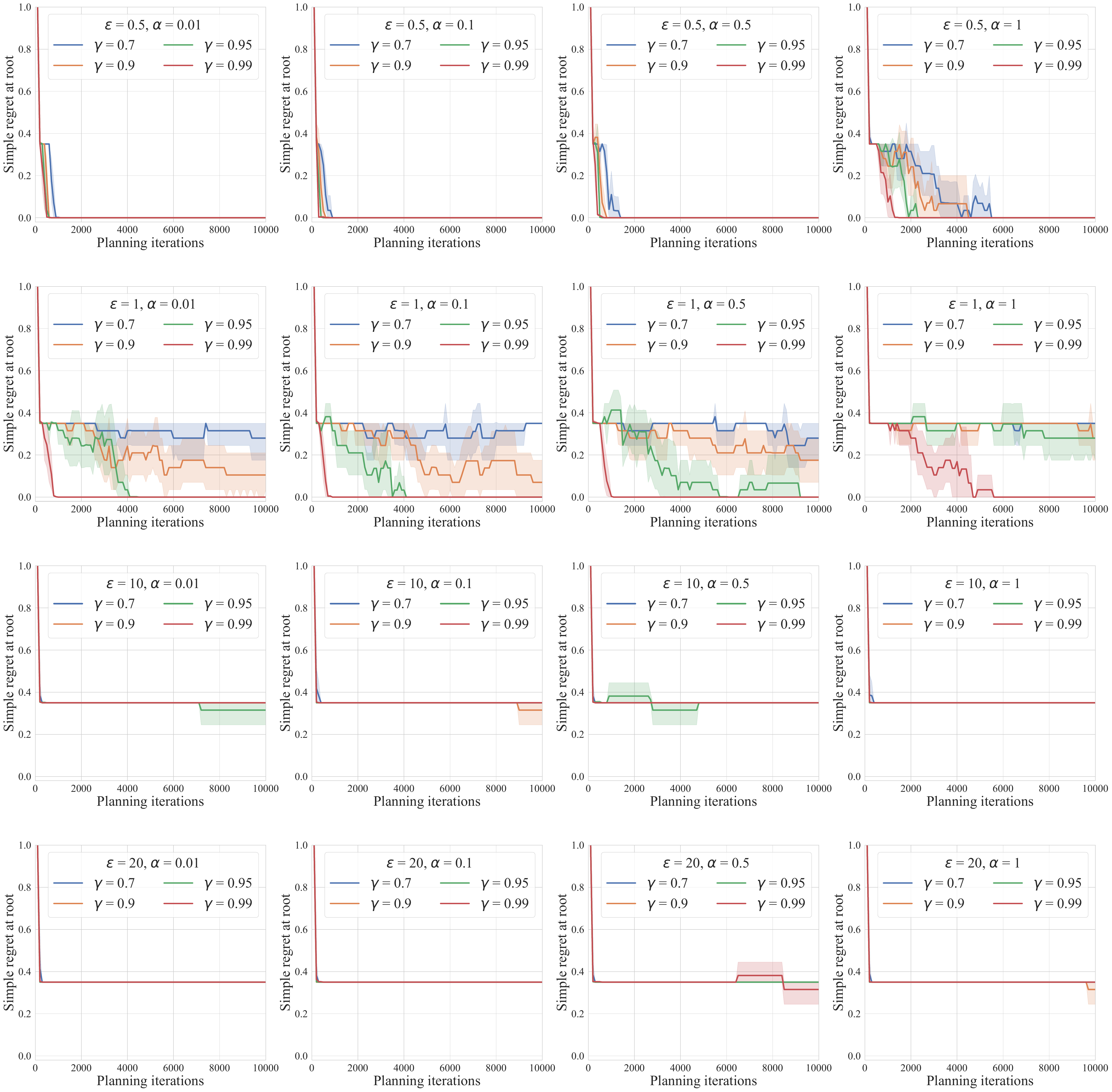}
		\caption{Results of FA-MCTS on the \emph{modified} D-chain problem with $D=20$ and 2 agents for varying exploration bias $\varepsilon$ and discount factor $\gamma$.
        \label{fig:decaying_20_modified_supp}\normalsize 
        }
	\end{center}
\end{figure*}
\clearpage
\begin{figure*}[p]
	\begin{center}
        \includegraphics[width=\linewidth]{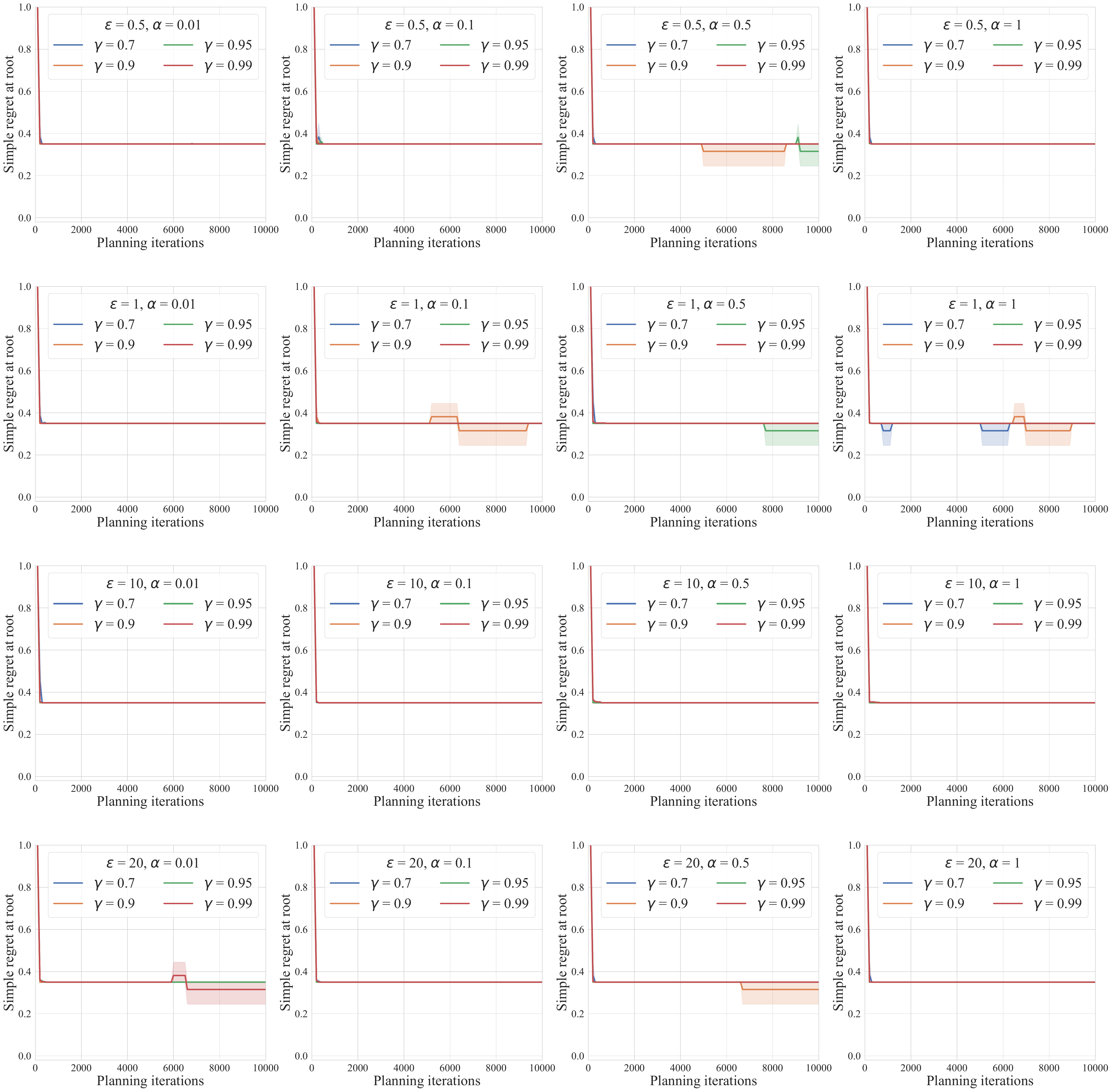}
		\caption{Results of NE-MCTS on the \emph{modified} D-chain problem with $D=20$ and 2 agents for varying exploration bias $\varepsilon$ and discount factor $\gamma$, and search temperature $\alpha$.
        \label{fig:noentropy_20_modified_supp}\normalsize 
        }
	\end{center}
\end{figure*}
\clearpage
\noindent\textbf{Frozen Lake Problem} We repeated the same experiments and hyperparameter search procedure as in the main paper to determine the optimal configurations for FA-MCTS. The selected parameter combinations, corresponding to those that achieved the highest cumulative joint score, are $\varepsilon = 0.5$, $\gamma = 0.9$, and $\alpha_{init} = 1$.

Figure~\ref{fig:frozenlake_supp} presents the performance of FA-MCTS in the Frozen Lake environment, alongside CB-MCTS and Dec-MCTS for reference. As expected, the faster decay rate reduces exploration, resulting in a $10\%$ decrease in the probability of reaching both goals compared with CB-MCTS. Nevertheless, FA-MCTS still achieves joint scores that closely match those of CB-MCTS and substantially outperform Dec-MCTS, highlighting its robustness even under more aggressive decay dynamics.

\begin{figure}[hb]
	\begin{center}
        \includegraphics[width=0.625\linewidth]{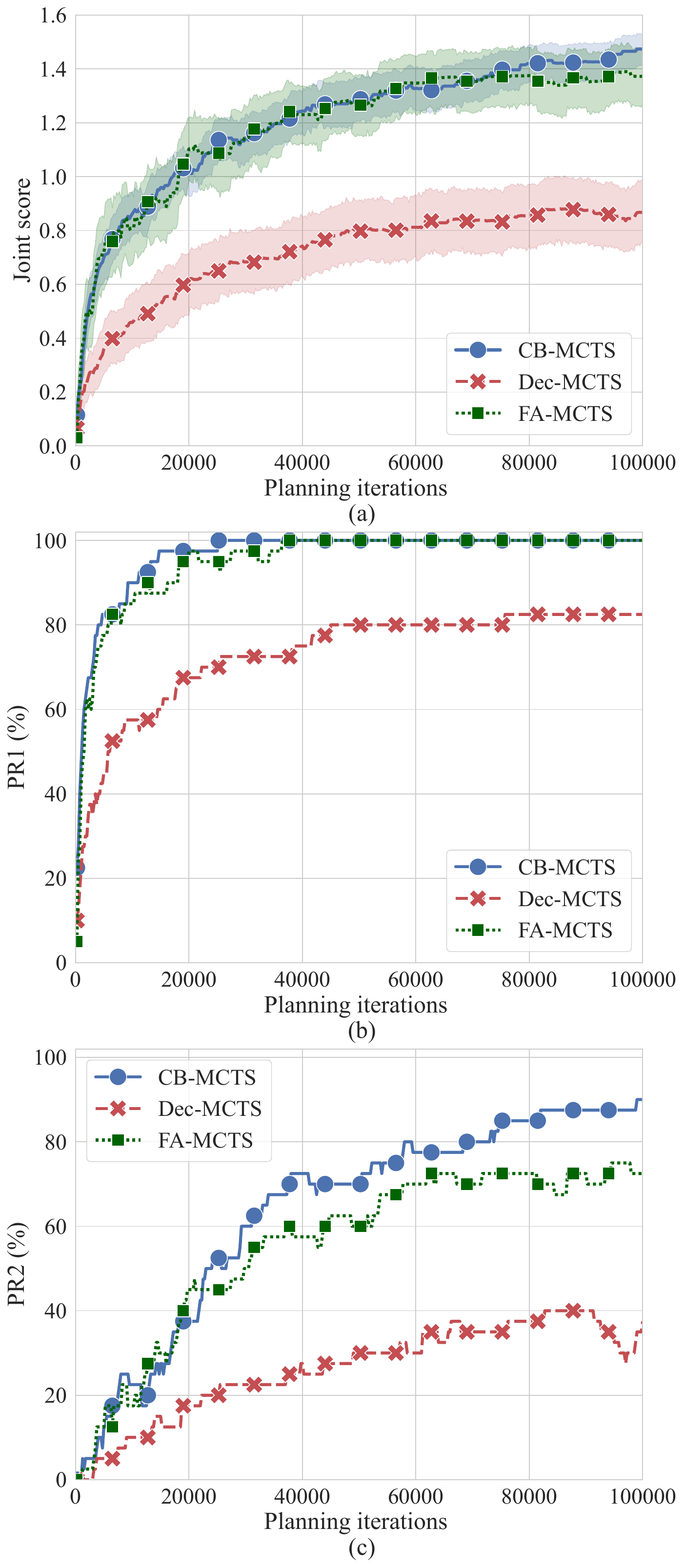}
		\caption{Performances of different algorithms in the Frozen Lake problem: Joint scores (a), Probability that at least one goal is reached by any agent (b), and Probability that both goals are reached (c). Results are with $95\%$ confidence interval.\label{fig:frozenlake_supp}\normalsize}
	\end{center}
\end{figure}

\noindent\textbf{Oil Rigs Inspection Problem} We follow the same experimental procedure as in the main paper, using the identical parameter settings of the full algorithm for FA-MCTS without additional tuning. Figure~\ref{fig:oilrig_supp} reports its performance on the Oil Rigs Inspection problem, compared with CB-MCTS and Dec-MCTS.

Consistent with previous findings, FA-MCTS achieves strong and stable performance, often yielding the best results across tested configurations. However, while FA-MCTS benefits from rapid convergence due to its faster temperature decay, the baseline CB-MCTS eventually surpasses it as the number of planning iterations grows. This contrast highlights a fundamental trade-off: FA-MCTS prioritizes stability and faster convergence by reducing exploration persistence, whereas CB-MCTS, with its slower decay search temperature, maintains greater exploratory flexibility that pays off in larger planning horizons by committing more effectively to high-reward regions of the joint search space.

These findings highlight the importance of tuning exploration dynamics when scaling to more complex multi-agent domains. Overall, they reinforce the high adaptability and robustness of our proposed algorithm as a decentralized planning solution in both smooth and sparse reward environments.

\begin{figure}[h]
    \begin{center}
        \includegraphics[width=0.625\linewidth]{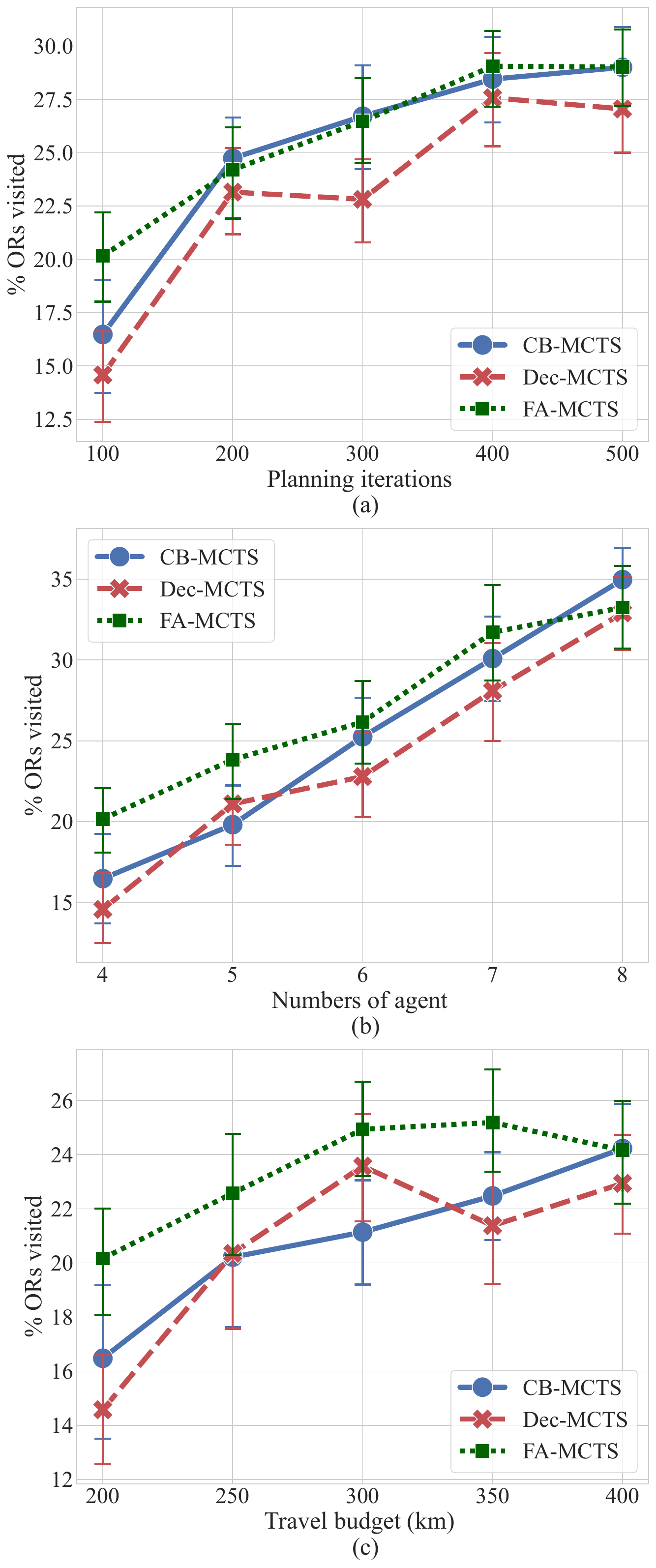}
        \caption{Impact of different system parameters on the algorithms' performances in the Oil Rigs Inspection problem: Planning iterations (a), Number of agents (b), and Travel budget (c). Results are with $95\%$ confidence interval.}
        \label{fig:oilrig_supp}
    \end{center}
\end{figure}

\end{document}